\pdfoutput=1
\documentclass[11pt, letterpaper]{article}
\usepackage[margin=1.2in,footskip=0.25in]{geometry}

\usepackage[T1]{fontenc}
\usepackage[scaled=0.8]{beramono}

\usepackage{amsmath, amssymb}
\usepackage{amsthm}
\usepackage{bm}
\allowdisplaybreaks

\usepackage[labelfont=bf]{caption}

\usepackage{xcolor}
\definecolor{linkcolor}{HTML}{0D5661} 
\usepackage{hyperref}
\hypersetup{
     colorlinks   = true,
     allcolors    = linkcolor,
     bookmarks    = false
}

\definecolor{white}{RGB}{255,255,255}
\definecolor{crimson}{RGB}{220,20,60}
\definecolor{blue}{RGB}{0,0,205}
\definecolor{myblue}{RGB}{80,80,160}
\definecolor{mygreen}{RGB}{80,160,80}
\colorlet{tcrimson}{white!40!crimson}
\colorlet{tblue}{white!40!blue}
\definecolor{porange}{HTML}{EE7F2D}
\definecolor{blue-plot}{HTML}{3A8FB7}

\def\*#1{\mathbf{#1}}

\newcommand{\argmax}{\mathop{\rm arg~max}\limits}
\newcommand{\argmin}{\mathop{\rm arg~min}\limits}

\newcommand{\E}{\mathbb{E}}

\newcommand{\R}{\mathbb{R}}

\newcommand{\hP}{\mathbb{P}}

\newcommand{\indep}{\!\perp\!\!\!\perp}

\theoremstyle{definition}
\newtheorem{proposition}{Proposition}

\newtheorem{lemma}{Lemma}
\newtheorem{assumption}{Assumption}

\numberwithin{equation}{section}

\usepackage{setspace}
\usepackage{pdfpages}
\usepackage{booktabs}
\usepackage{titlesec}

\usepackage{natbib}
\begin{document}
\title{{\Large\textbf{
  Difference-in-Differences for Ordinal Outcomes:\\
Application to the Effect of Mass Shootings on Attitudes toward Gun Control}}\footnote{I
 am grateful to Matt Blackwell, Gary King, Kosuke Imai, Molly Offer-Westort, Ikuma Ogura, Shun Yamaya, members of Imai research group at Harvard
 (Soubhik Barari, Jake Brown, Naoki Egami, Shusei Eshima, Max Goplerud, June Hwang, Connor Jerzak,
  Shiro Kuriwaki, Santiago Olivella, Sun Young Park, Casey Petroff, Avery Schmidt, Sooahn Shin,
	Tyler Simko and Diana M. Stanescu) and participants of G3 Mini-Conference
 for comments and suggestions.
 The \textsf{R} package \texttt{orddid} is available for implementing the proposed methodology at \url{https://github.com/soichiroy/orddid}.}
}
\author{{\large Soichiro Yamauchi}\footnote{Ph.D. Candidate, Department of Government and Institute for Quantitative Social Science, Harvard University. Email: \href{syamauchi@g.harvard.edu}{syamauchi@g.harvard.edu}. URL: \url{https://soichiroy.github.io/}.}}
\date{{\normalsize This version: \today}\\
      {\normalsize First draft: October 29, 2019}}
\maketitle

\setstretch{1.2}
\begin{abstract}
  The difference-in-differences (DID) design is  widely used in observational studies to estimate the causal effect of a treatment when repeated observations over time are available.
  Yet, almost all existing methods assume linearity in the potential outcome (\emph{parallel trends} assumption) and target the additive effect.
  In social science research, however, many outcomes of interest are measured on an ordinal scale.
  This makes the linearity assumption inappropriate because the difference between two ordinal potential outcomes is not well defined.
  In this paper, I propose a method to draw causal inferences for ordinal outcomes under the DID design.
  Unlike existing methods,
  the proposed method utilizes the latent variable framework to handle the non-numeric nature of the outcome,
  enabling identification and estimation of causal effects
  based on the assumption on the quantile of the latent continuous variable.
  The paper also proposes an equivalence-based test to assess the plausibility of the key identification assumption when additional pre-treatment periods are available.
  The proposed method is applied to a study estimating the causal effect
  of mass shootings on the public's support for gun control.
  I find little evidence for a uniform shift toward pro-gun control policies as found in the previous study, but find that the effect is concentrated on left-leaning respondents who experienced the shooting for the first time in more than a decade. 
\end{abstract}

\clearpage
\setstretch{1.35}

\section{Introduction}\label{sec:introduction}


The difference-in-differences (DID) design is widely used in observational studies with repeated observations over time \citep{card1994minimum,angrist2008mostly,lechner2011estimation}.
It allows scholars to identify the causal effect accounting for time-invariant unobserved confounders.
Although significant progress has been made to improve the original DID design in recent years \citep[e.g.,][]{abadie2005semiparametric,athey2006identification,qin2008empirical,lee2016generalized,
arkhangelsky2018synthetic,callaway2018difference,
li2019double,lu2019robust},
most of the existing methods attempt to identify and estimate the treatment effect under the linearity assumption \citep{abadie2005semiparametric}.
This parallel-trends assumption imposes a restriction on the potential outcomes
such that
the mean of the treatment and the control group has identical trends in the absence
of the treatment.
Therefore, the assumption is meaningful only when the difference between two potential outcomes is well defined (e.g., continous outcomes).

In social science research, however, many outcomes of interest are measured on an ordinal scale.
For example, in political science, scholars measure voter's ideology on a scale from ``very liberal'' to ``very conservative''
\citep[e.g.,][]{gay2002spirals, jessee2016can, mason2015disrespectfully}
or ask an attitude toward a policy item from ``strongly disagree'' to ``strongly agree''
\citep[e.g.,][]{grose2015explaining, frymer2020labor, likert1932technique}.
In fact, due to the limitation of space and other administrative reasons,
most of the questions asked in major social science surveys are ordinal.
When the outcome is measured on such a scale,
it is difficult to define the ``mean'' of non-numeric variables and further impose a restriction on their ``differences.''
In addition,
the usual definition of the treatment effect as the difference between two potential outcomes is not well defined \citep[e.g.,][]{volfovsky2015causal,lu2018treatment}, unless strong assumptions about the scale are imposed.
This implies that the standard DID cannot be  directly used for ordinal outcomes.

With a dearth of methods tailored for analyzing ordinal outcomes in the DID setting, scholars often treat them as continuous, dichotomize them with some threshold,
or employ the ordered logistic (probit) regression.
Each of the three approaches has its own shortcomings.
By treating the ordinal outcome as continuous, scholars implicitly assume that
categories are equally spaced. This assumption is not testable nor appropriate in many applications.
Although dichotomizing the outcome appears to enable scholars to adopt the standard DID method to estimate casual effects,
this strategy is not robust to different transformations (i.e., different choices of the dichotomization threshold).
Specifically, due to the non-linear nature of the ordinal outcome,
the validity of the parallel trends assumption under one transformation
does not guarantee the validity of the assumption under another transformation.
This is problematic because oftentimes scholars do not have substantive justification on which transformation should be employed.

In  this paper, I develop a methodology for estimating causal effects for ordinal outcomes with repeated observations.
Instead of assuming linearity on the actual outcome, I utilize the latent variable formation often used in the discrete choice models.
Because the assumptions are imposed on the entire distribution of the latent variables,
the proposed method does not require a transformation of the outcome variable.
Furthermore, the method enables researchers to estimate interpretable causal effects,
defined as a difference between two probabilities,
under a single set of assumptions.
I also propose a diagnostic tool when scholars have data from more than one pre-treatment period.
As in the standard DID for continuous outcomes, where scholars can check if the pre-treatment trends are parallel,
this diagnostic test allows researchers to formally confirm whether the assumption holds at least during the pre-treatment periods.

The method of this paper is closely connected to the literature on non-linear DID
\citep[e.g.,][]{athey2006identification, sofer2016negative, callaway2018quantile, glynn2019generalized}.
In particular, \cite{athey2006identification} consider an extension of their method
to binary and count outcomes, but they do not consider the case of ordinal outcomes.
Most importantly, because they impose minimal restrictions on the potential outcome, their method does not provide point identification even for an additive effect and the bound can be uninformative.
Instead, I impose a stronger assumption for the sake of identifying the non-additive causal effect,
which enables researches to estimate informative causal effects.



The proposed methodology is used to revisit a recent debate on the effect of mass shootings on attitudes towards gun control regulations \citep{barney2019reexamining,hartman2019accounting,newman2019mass}.
In their original and follow-up studies, Hartman and Newman
argue that a mass shooting increases people's support for stricter gun control policies
regardless of their partisanship.
They also argue that the effect is conditional on geographical context, such as
the safety of their neighborhood.
On the other hand, Barney and Schaffner argue that
there is no strong evidence to support the claim of Hartman and Newman.
They also find a polarizing effect of mass shootings where Democrats become more supportive of gun control
while Republicans become less supportive of gun control.

In Section~\ref{sec:application}, I re-analyze the data from the motivating empirical study using the proposed method.
Using two-wave panel data,
I find that mass shootings have an effect on those who experience mass shootings for the first time in a decade: they form a stronger opinion that supports gun control regulations,
while the result suggests little evidence for a uniform shift toward pro-gun control policies.
I also find that the effect is concentrated among Democrats including those who weakly identify themselves
as Democrat.
However, the effects among Republicans are not statistically distinguishable from zero. Thus, I find little evidence to support the polarizing effect of mass shootings.
Using three-wave panel data, which provides an additional pre-treatment time period, I assess the plausibility of the identification assumption.
The proposed testing procedure finds a supportive evidence for the validity of the assumption. 
Reanalysis of the three-wave panel, however, finds little evidence to support the claim
that mass shootings have any effect on the support for gun control.

The rest of the paper is organized as follows.
Section~\ref{sec:application-introduction} introduces the motivating application of the method.
In Section~\ref{sec:methodology}, the methodology is introduced
where I discuss identification assumptions and estimation strategy.
In Section~\ref{sec:application}, I apply the proposed method to the data described in Section~\ref{sec:application-introduction}.
Finally, I offer concluding remarks in Section~\ref{sec:conclusion}.

\section{The Effect of Mass Shootings on Public Support for Gun Control}\label{sec:application-introduction}


\subsection{The debate on the effect of mass shootings}
This section describes the design of observational studies that
investigate the effects of an event on an ordinal outcome measured over time.
\cite{newman2019mass} and the follow-up studies \citep{barney2019reexamining,hartman2019accounting} study the relationship between experiencing mass shootings and the attitude to gun control.
These studies use survey data with a two-wave and three-wave panel
to investigate whether living in close proximity to mass public shootings has a causal impact on people's attitude to gun control.
Respondents to the survey are considered as ``treated''
if at least one mass shooting occurs within 100 miles from their residential zip code.
To measure the attitude toward gun control, the authors used a response to the following survey question in the Cooperative Congressional Election Study (CCES) \citep{DVN/II2DB6_2018, DVN/TOE8I1_2015}:
\begin{quote}
  \texttt{
  In general, do you feel that laws covering the sale of firearms should be made more strict, less strict, or kept as they are?
  }

  \texttt{
  (0) Less Strict;
  (1) Kept As They Are;
  (2) More Strict.
  }
\end{quote}

The original studies utilize variations in treatment assignment over time to isolate the effect of mass shootings from the time trends and location effects.
Based on the analysis, Hartman and Newman find that living in near proximity to mass public shootings moves people to support stricter regulations on gun sales \citep{newman2019mass, hartman2019accounting}.
They also report that the effect does not vary by respondents' party affiliation.
In a follow-up study, \cite{barney2019reexamining} correct data and conduct additional analyses.
They conclude that the effect varies by which party people affiliate with
and that there is a polarizing effect of mass shootings.
Democrats become \emph{more} supportive of stricter gun control,
while Republican become \emph{less} supportive of regulations.

Throughout the debate, the authors utilize a variety of methodologies, such as
a ordered probit model with random effects and a linear fixed effect model,
to estimate the impact of mass shootings on the attitude (see Table~\ref{tab:summary-method} in Appendix~\ref{sec:detail-application}).
Although ``difference-in-differences'' is  mentioned in these studies,
discussions about the quantities of interest and assumptions required for the identification of those quantities are missing from the debate.
Without assessing the assumptions explicitly, it is challenging to conclude whether mass shootings have any effect on people's attitude.
This paper fills this gap by proposing a methodology
that enables researches to assess assumptions and reliably estimate causal effects.

\subsection{The methodological challenge}
For a causal effect to be reliably estimated,
assumptions must be imposed and evaluated.
In the following, I demonstrate that even if researchers are aware of identification assumptions,
the current practice is not well suited for the ordinal outcome.
Suppose that, following common practice, researchers dichotomize an ordinal outcome into a binary outcome by specifying some threshold.
The standard DID analysis is then applied on this transformed outcome.
In our running example, there are two possible ways to transform
the original outcome into a binary variable.
One way is to code \texttt{more-strict} category as one and \texttt{kept-as-they-are} and \texttt{less-strict} as zero; the other way is to treat \texttt{more-strict} and \texttt{kept-as-they-are}
as one and \texttt{less-strict} as zero.

\begin{figure}[htb]
  \centerline{\includegraphics[scale=0.8]{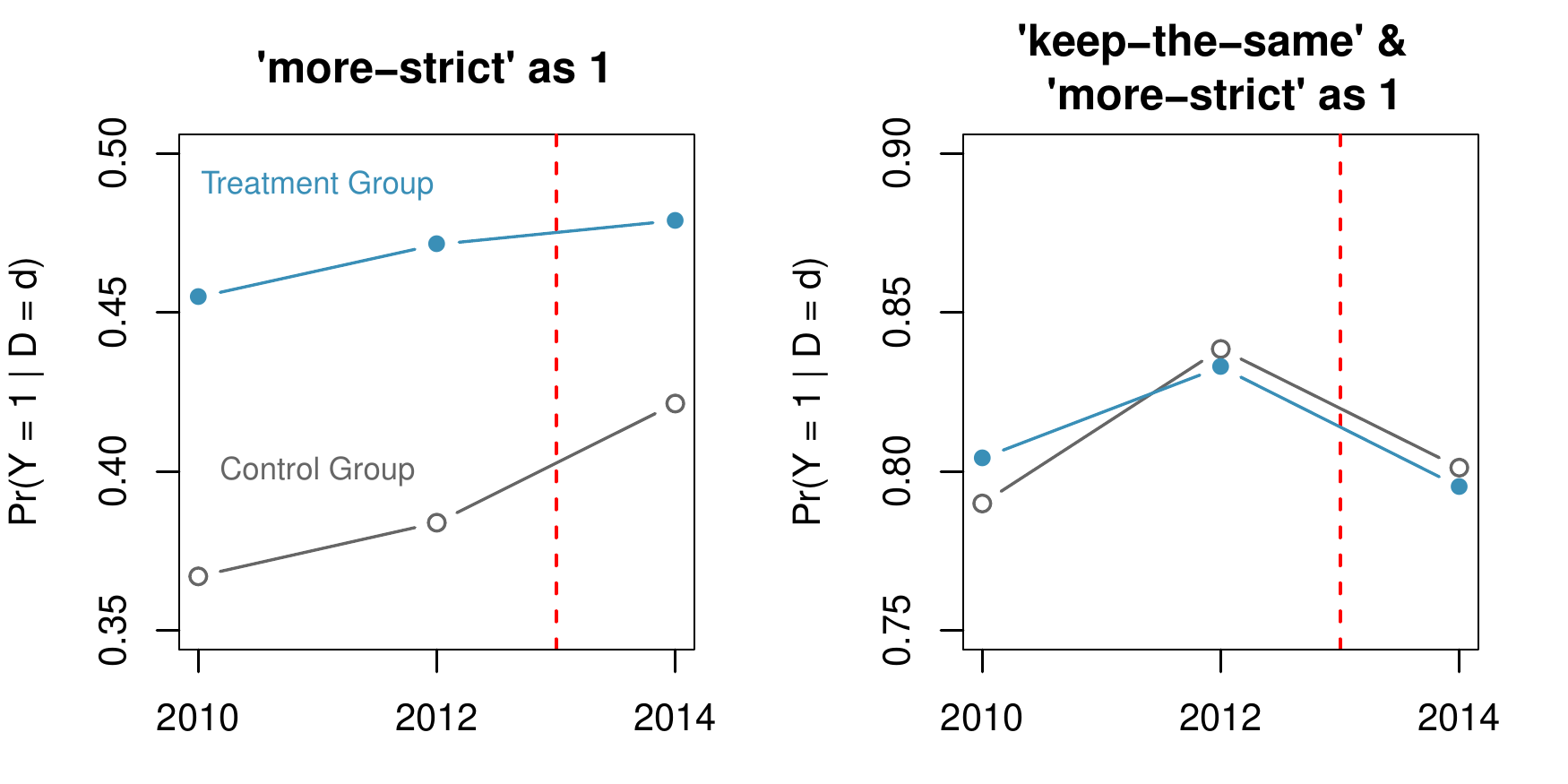}}
  \caption{Visual assessment of parallel trends assumption for the three-wave panel from the survey data.
  	Respondents who are not treated until 2012 are used for generating this plot.
    The lines with solid circles (\textcolor{blue-plot}{blue}) show trends for the treated group and the lines with hollow circles (\textcolor{gray}{gray}) show trends for the control group.
    Vertical dashed lines (\textcolor{red}{red}) show the timing of the treatment.
    Although the left panel appears to show that pre-treatment trends are parallel between the treatment and the control group, the right panel suggests that pre-treatment trends are not parallel.
    }
  \label{fig:pt-assess-gun}
\end{figure}

After transforming the outcome,
scholars can check if the pre-treatment trends are parallel for this new binary variable.
Inspecting pre-treatment trends is a routine often used in empirical studies to
justify the use of the DID design \citep{angrist2008mostly}.
Figure~\ref{fig:pt-assess-gun} shows trends for each transformed outcome using the three-wave panel of the survey
where I subset respondents who are not treated until 2012.
The panel on the left shows the first type of transformation
where only $\texttt{more-strict}$ category is coded as $1$ (denoted by $Y$ on the \textit{y}-axis).
We can see that the pre-treatment trends between the treatment group (\textcolor{blue-plot}{blue}) and the control group (\textcolor{gray}{gray}) appears to be parallel (denoted by $D$ on the \textit{y}-axis.).
Thus if a researcher transforms the original variable in this way,
she might conclude that the DID design is suitable for analyzing the data.
The panel on the right shows the pre-treatment trends
for the second type of transformation.
In this case, however,
the pre-treatment trends do not seem to be parallel:
the trends cross during the pre-treatment period.

Note that this observation is not specific to this application.
In Appendix~\ref{appendix:binarize-outcome}, I demonstrate that it is trivial to construct an example that satisfying parallel trends in one transformation does not imply the parallel trends in another transformation.

It is often unclear \textit{ex ante} which threshold should be chosen from a substantive point of view.
Therefore, it is unfortunate that the validity of the design appearently depends on how the variable is transformed.
Although the running example only has three categories, the problem exacerbates when scholars need to analyze an outcome that has a larger number of categories.

\section{The Proposed Methodology}\label{sec:methodology}


\subsection{The setup}
Let $Y_{it} \in \{0, \ldots, J-1\} \equiv \mathcal{J}$ denote the observed outcome measured on an ordinal scale with $J$ categories ($J \geq 3$)
for unit $i \in \{1, \ldots, n\}$ and time $t \in \{0, 1\}$.
The binary treatment, denoted by $D_{i} \in \{0,1\}$, is assigned after $Y_{i0}$ is observed but before time $t = 1$.
We use the potential outcome notation to denote the counterfactual outcome, $Y_{it}(d)$ for $d \in \{0, 1\}$.
For example, $Y_{i1}(0)$ is an attitude toward gun control regulations that would realize in the post-period if a respondent did not experience a mass shooting (i.e., the control condition).

In many applications, scholars are interested in estimating the distributional treatment effect.
In this paper, I focus on the treatment effect on the treated.
Specifically, the effect $\zeta_{j}$ is defined as the difference in probabilities of choosing category $j$ under two conditions,
\begin{equation}\label{eq:dist-effect}
\zeta_{j} = \hP(Y_{i1}(1) = j \mid D_{i} = 1) - \hP(Y_{i1}(0) = j \mid D_{i} = 1).
\end{equation}
for $j \in \mathcal{J}$.
In our application, $\zeta_{2}$ is the difference in probabilities that those treated prefer more strict gun control between the treated and the control conditions.
Thus, observing $\zeta_{2} > 0$ implies that the mass shootings make people prefer stricter policies on gun control.
Similarly, $\zeta_{0}$ is the effect of the treatment on \texttt{less-strict} category
and $\zeta_{0} > 0$ implies that incidents turn people to prefer less strict regulations.

When the number of categories is large, it is sometimes useful to estimate the cumulative effect $\Delta_{j}$, which is defined as a difference in probabilities of choosing $j$ or larger
categories under the two conditions,
\begin{equation}\label{eq:dist-cum-effect}
  \Delta_{j} = \hP(Y_{i1}(1) \geq j \mid D_{i} = 1) - \hP(Y_{i1}(0) \geq j \mid D_{i} = 1)
\end{equation}
for $j \in \mathcal{J}\backslash \{0\}$.
Note that $\Delta_{j} = \sum^{J-1}_{\ell = j}\zeta_{\ell}$ by construction,
and thus it is sufficient to consider the identification of $\zeta_{j}$.

The cumulative effect is also useful to connect the approach that dichotomizes ordinal outcomes
to the proposed method.
From the above definition, we can see that the standard DID based on the dichotomized outcome at threshold $j$ identifies $\Delta_{j}$.
This is because $\Pr(Y_{i1}(d) \geq j \mid D_{i} = 1) = \E[\mathbf{1}\{Y_{i1}(d) \geq j\} \mid D_{i} = 1] \equiv \E[\widetilde{Y}_{i1}(d) \mid D_{i} = 1]$
where $\widetilde{Y}_{i1}(d)$ is the dichotomized potential outcome with threshold $j$.
This means that the standard DID applied to the dichotomized outcome can estimate only one of $J-1$ possible quantities of interest.
Furthermore, if one wishes to estimate all possible $\Delta_{j}$'s by changing the threshold, it requires $J-1$ distinct identification assumptions.
As we saw in our motivating example, however, satisfying the parallel trends assumption for $\Delta_{j}$ does not necessarily imply that the assumption for $\Delta_{j'}$ is satisfied.


\subsection{Identification}
Typically, we do not have a good sense of which estimand is best suited for answering the substantive question.
Therefore, it is natural that we attempt to identify and estimate $\zeta_{j}$ for all $j \in \mathcal{J}$
from the observed data.
The goal of this section is to establish the identification of $\bm{\zeta} = (\zeta_{0}, \ldots, \zeta_{J-1})^{\top}$ with a single set of assumptions.

To compute the quantity defined in Equation~\eqref{eq:dist-effect},
we need the marginal distribution of $Y_{i1}(1)$ and $Y_{i1}(0)$ for the treated.
While we observe $Y_{i1}(1)$ for $D_{i} = 1$ because $Y_{i1} = D_{i}Y_{i1}(1) + (1 - D_{i})Y_{i1}(0)$,
we need to impose additional assumptions to identify the distribution of $Y_{i1}(0)$ for $D_{i} = 1$.
Following \cite{athey2006identification}, I omit the subscript $i$ for units and denote
$Y_{dt} \sim Y_{it}(0)\mid D_{i}= d$
where $A \sim B$ indicates $A$ and $B$ are equivalent in distribution.
$Y_{dt}$ denotes the potential outcome under the control condition
at time $t$ for group defined by $D_{i} = d$.
While we observe $Y_{00}$, $Y_{01}$ and $Y_{10}$,
the counterfactual outcome $Y_{11} \sim Y_{i1}(0) \mid D_{i} = 1$ is  what we do not observe in the data.
In our example, $Y_{11}$ is the potential attitude to gun control that we would have observed if
those respondents who have experienced mass shootings would have not been exposed to the event.

I first impose a structure on the potential outcome.
Specifically, I assume that the observed categorical outcome follows the index model,
which means that there is a latent variable behind $Y_{dt}$
and that the categorical outcome is defined by a simple thresholding rule on the latent variable.
\begin{assumption}[Index model]\label{assumption:model}
Assume that the potential outcomes follow the index model
such that
there exists a latent variable $Y^{*}_{dt} \in \R$ and
\begin{equation}
  Y_{dt}  =
  \begin{cases}
  0      &\text{if}\quad \kappa_{0} \leq Y^{*}_{dt} < \kappa_{1}   \\
  j      &\text{if}\quad \kappa_{j} \leq Y^{*}_{dt} < \kappa_{j+1} \\
  J-1    &\text{if}\quad \kappa_{J-1} \leq Y^{*}_{dt} \leq \kappa_{J}\\
  \end{cases}
\end{equation}
where $\{\kappa_{j}\}^{J}_{j=0}$ are a set of cutoffs
with $\kappa_{0} = -\infty$ and $\kappa_{J} = \infty$.
\end{assumption}

Assumption~\ref{assumption:model} says that the potential outcome
defined on an ordinal scale $Y_{dt}$ is a function of another potential outcome
defined on a continuous space $Y^{*}_{dt}$.
In the application, $Y^{*}_{dt}$ can be considered as the underlying intensity of one's attitude toward gun control policies
where larger value of $Y^{*}_{dt}$ corresponds to a support for stricter gun control.
The assumption allows us to handle the outcome on a continuous space through $Y^{*}_{dt}$
instead of directly working on a discrete space.
Note that $\kappa_{j}$'s are constants assumed to be fixed and they do not depend on group ($d$) nor time ($t$).

Different from the additive effect, the distributional treatment effect $\zeta_{j}$ requires that the entire  marginal distribution of the potential outcome is identified.
For that,
I further impose a distributional assumption on $Y^{*}_{dt}$ in Assumption~\ref{assumption:distribution}.
\begin{assumption}[Location-scale family assumption]\label{assumption:distribution}
Let $U$ denote a continuously distributed random variable with mean $0$ and variance $1$
that belongs to a parametric family.
We assume that $Y^{*}_{dt}$ belongs to the location-scale family,
that is, it can  be written as
\begin{equation}
Y^{*}_{dt} \sim \mu_{dt} + \sigma_{dt}U
\end{equation}
where $\mu_{dt}$ is the location and $\sigma_{dt}$ is the scale parameter.
\end{assumption}

Assumption~\ref{assumption:distribution} specifies the distribution of the latent utilities.
It assumes that each marginal distribution belongs to the location-scale family distribution with time and group specific location and scale parameter.
This implies that the distribution of the potential outcomes are different up to mean and the scale.
Note that the joint distribution of the latent utilities are left unspecified,
so units can have correlated latent utilities over time.
Although this is a parametric assumption (i.e., the distribution of $U$ should be known),
the location-scale family encompasses a large class of parametric distributions (e.g., the normal distribution, the logistic distribution or the \textit{t}-distribution, etc).

Finally, I impose a structure on the relationship between latent variables $Y^{*}_{dt}$.
This allows us to map what we observe in the control group over time
to what would have happened to the treated group if it was not treated.
I first start with a restrictive assumption that is similar to the standard DID design.
It is possible to assume that the parallel trends hold on the latent outcome, that is,
\begin{equation}
\E[Y^{*}_{i1} \mid D_{i} = 1]  - \E[Y^{*}_{i0}\mid D_{i} = 1]
=
\E[Y^{*}_{i1} \mid D_{i} = 0]  - \E[Y^{*}_{i0}\mid D_{i} = 0].
\end{equation}
Then, the mean of the counterfactual latent outcome $Y^{*}_{11}$ is uniquely identified as
\begin{equation*}
\mu_{11} = \mu_{10} + \mu_{01} - \mu_{00}.
\end{equation*}

However, this approach is restrictive, because it requires an additional assumption
that the variance is constant across time and groups, that is, $\sigma_{dt} = \sigma$ for all $d$ and $t$; otherwise we cannot identify the entire distribution of the latent outcome $Y^{*}_{11}$.
This constant variance assumption is strong because it only allows the unidirectional change of choice probabilities.

Therefore, I impose a different assumption from the standard parallel trends assumption.
Instead assuming the mean shift, the assumption is imposed on the entire distributions,
which is originally introduced by \cite{athey2006identification} \citep[also see][]{sofer2016negative}.
Specifically, I assume that the shift in the distribution across time are constant between the treatment and the control groups.
Figure~\ref{fig:CiC-assumption} graphically illustrates the assumption.
The key part of this assumption is that
the vertical arrows in the two graphs should be the same length.
In other words, $q_{d}(v) - v$ captures the trend in the distribution (i.e., how much $Y^{*}_{dt}$ ``shifts'' between $t = 0$ and $t = 1$) and the assumption says that the ``shift'' is identical across two groups.
This means that for each choice of $v$, the corresponding value of $q_{d}(v)$ (``shift'') should be the same for $d = 0, 1$.
\begin{figure}[htb]
  \centerline{\includegraphics[width=\textwidth]{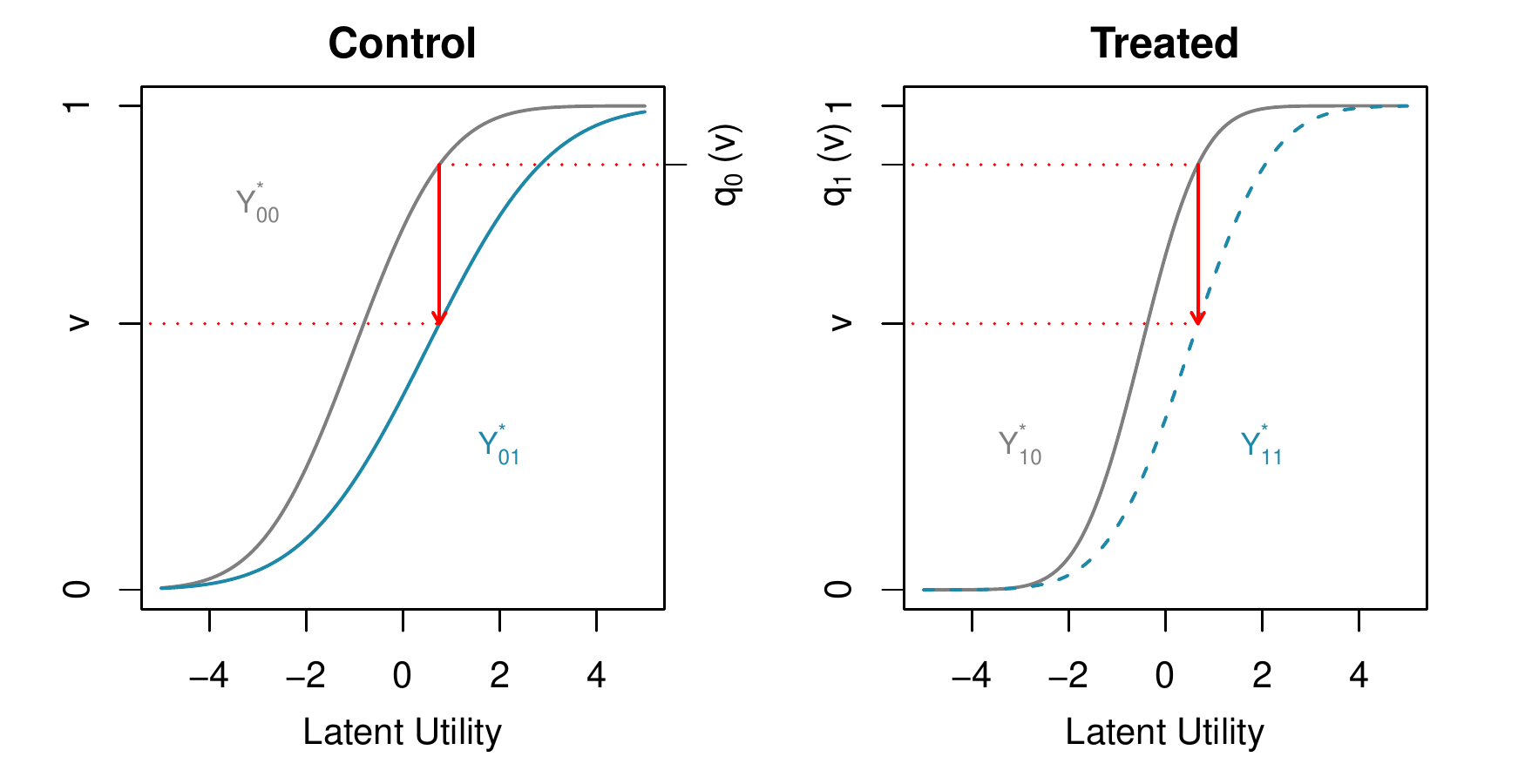}}
  \caption{Graphical illustration of Assumption~\ref{assumption:quantile}. Left (right): cumulative distribution functions of the latent utilities $Y^{*}_{dt}$ under the control (treatment) condition.  Blue (gray) lines indicate the distribution for time $t = 1$ ($t = 0$). Dashed line on the right panel is the distribution of counterfactual outcome $Y^{*}_{11}\sim Y^{*}_{i1}(0)|D_{i} = 1$.
  The key assumption is that the length of the vertical arrow (\textcolor{red}{red}) is the same between the two panels for all range of $v$. This allows us to recover the shape of the dashed line based on latent utility distributions for the observed outcomes (i.e., solid lines).}
  \label{fig:CiC-assumption}
\end{figure}

Assumption~\ref{assumption:quantile} formally states the assumption.
\begin{assumption}[Distributional parallel trends \citep{athey2006identification}]\label{assumption:quantile}
Let $F_{Y^{*}_{dt}}(y) = \hP(Y^{*}_{dt} \leq y)$ be the cumulative distribution function (CDF) of $Y^{*}_{dt}$
and define $q_{d}(v) = F_{Y^{*}_{d0}} \circ F^{-1}_{Y^{*}_{d1}}(v)$.
Then, we assume that for all $v \in [0, 1]$,
\begin{equation}
q_{1}(v) = q_{0}(v)
\end{equation}
\end{assumption}
Assumption~\ref{assumption:quantile} imposes a restriction on the relationship between
the pre-treatment latent outcome $Y^{*}_{10}$ and the counterfactual latent outcome $Y^{*}_{11}$,
based on the relationship between two latent variables in the control group.
Note that by construction $q_{1}(v) - q_{0}(v) = 0$ for $v = 0,1$
because CDFs should agree at the end of the support, $\lim_{y\to\pm\infty}F_{Y^{*}_{d0}}(y) = \lim_{y\to\pm\infty}F_{Y^{*}_{d1}}(y)$.

Assumption~\ref{assumption:model}, \ref{assumption:distribution}
and \ref{assumption:quantile}
identify the distribution of the counterfactual outcome.
Proposition~\ref{proposition:identification} presents the formal result.

\begin{proposition}[Identification of the Counterfactual Distribution]\label{proposition:identification}
Under Assumption~\ref{assumption:model}, \ref{assumption:distribution},
and  \ref{assumption:quantile},
the distribution of the counterfactual latent utility $Y^{*}_{11}$ is identified as
\begin{equation}
  Y^{*}_{11} \sim \mu_{11} + \sigma_{11}U
\end{equation}
where
\begin{align*}
\mu_{11} = \mu_{10} + \frac{\mu_{01} - \mu_{00}}{\sigma_{00}/\sigma_{10}}
\quad\text{and}\quad
\sigma_{11} = \frac{\sigma_{10}\sigma_{01}}{\sigma_{00}}.
\end{align*}
And thus, the distribution of the potential outcome is identified as
\begin{equation*}
\hP(Y_{i1}(0) = j\mid D_{i} = 1) =
F_{U}\bigg(\frac{\kappa_{j+1} - \mu_{11}}{\sigma_{11}}\bigg)
-
F_{U}\bigg(\frac{\kappa_{j} - \mu_{11}}{\sigma_{11}}\bigg)
\end{equation*}
for $j = 0, \ldots, J - 1$,
where $F_{U}(u) = \hP(U \leq u)$ is the CDF of $U$.
\end{proposition}

A proof is in Appendix~\ref{appendix:proofs}.
Proposition~\ref{proposition:identification} says that
the location and the scale of the counterfactual latent outcome $Y^{*}_{11}$
are uniquely determined by parameters of observed outcomes.
This implies that we can recover the distribution of the counterfactual outcome $Y_{11} \sim Y_{i1}(0) | D_{i} = 1$ (i.e., the potential outcome under the control condition for the treated unit at time $t = 1$) using
parameters estimated from the observed data, $Y_{00}$, $Y_{01}$ and $Y_{10}$.
For example, if we assume that $U$ follows the standard normal distribution,
we have that $Y^{*}_{11}$ follows the normal distribution with mean $\mu_{11}$ and
variance $\sigma^{2}_{11}$.




\subsection{Estimation}\label{subsec:estimation}
The identification result in the previous section provides
a guidance on how we can estimate the causal effect from the observed data.
Let $\bm{\theta}_{dt} = (\mu_{dt}, \sigma_{dt})^{\top}$ denote a vector of parameters that characterize
the distribution of the latent utility $Y^{*}_{dt}$.
I take a two-step approach to estimate causal quantity $\zeta_{j}$ defined in Equation~\ref{eq:dist-effect}
for all $j \in \mathcal{J}$.
In the first step, I estimate parameters for observed outcomes,
that is, $\bm{\theta}_{00}$, $\bm{\theta}_{01}$ and $\bm{\theta}_{10}$.
Based on the estimate of these parameters, causal effects are estimated in the second step.

In this and the following section, I focus on a case where $Y^{*}_{dt}$ follows the normal distribution,
that is $U \sim \mathcal{N}(0, 1)$.
Then, by Assumption~\ref{assumption:distribution},
the observed outcomes $Y_{00}$, $Y_{01}$ and $Y_{10}$
follow the ordered probit model.
Thus, parameters $\bm{\theta} = (\bm{\theta}^{\top}_{00}, \bm{\theta}^{\top}_{01}, \bm{\theta}^{\top}_{10}, \bm{\kappa}^{\top})^{\top}$ can be estimated via the maximum likelihood.
  \begin{equation*}
  \widehat{\bm{\theta}} =
  \argmin_{\bm{\mu}, \bm{\sigma}, \bm{\kappa}} \sum^{n}_{i=1}\sum_{t\in \{0, 1\}}\sum_{j \in \mathcal{J}}\mathbf{1}\{Y_{it} = j, t D_{i} = 0\} \log \Big\{
  \Phi[(\kappa_{j+1} - \mu_{D_{i}, t})/\sigma_{D_{i}, t}]
  - \Phi[(\kappa_{j} - \mu_{D_{i}, t})/\sigma_{D_{i}, t}]
  \Big\}
  \end{equation*}
where $\mathbf{1}\{\cdot\}$ is an indicator function that takes $1$ if the argument is true and takes $0$ otherwise,
and $\Phi(\cdot)$ is the cumulative distribution function of the standard normal distribution.
Different from the standard ordered probit specification,
I fix two cutoffs $\kappa_{1}$ and $\kappa_{2}$ (recall that $\kappa_{0} = -\infty$ and $\kappa_{J} = \infty$).
This allows us to estimate the variance component in addition to means \cite[Lemma~\ref{lemma:latent-identification} in Appendix~\ref{appendix:proofs}; also see for example][Chapter 8]{jackman2009bayesian}.
Note that the choice of $\kappa$ is not consequential in that, the causal effect estimate $\widehat{\bm{\zeta}}$ is invariant to the choice of the cutoffs (Lemma~\ref{lemma:invariance-cutoff} in Appendix~\ref{appendix:proofs}).
This is because the identification assumption imposes a structure on the quantile scale,
which is invariant to the scale of the latent variables,
while different choices of cutoffs only affect the location and the scale (i.e., $\mu$ and $\sigma$) of the latent variables.

We then estimate the parameter for the counterfactual distribution $\bm{\theta}_{11} = (\mu_{11}, \sigma_{11})^{\top}$ by the plug-in estimator based on the first stage,
\begin{align}\label{eq:second-stage-theta}
\widehat{\mu}_{11} = \widehat{\mu}_{10} + (\widehat{\mu}_{01} - \widehat{\mu}_{00}) / (\widehat{\sigma}_{00}/\widehat{\sigma}_{10}),\quad
\text{and}\quad
\widehat{\sigma}_{11} = (\widehat{\sigma}_{10}\widehat{\sigma}_{01}) / \widehat{\sigma}_{00}.
\end{align}
Since the causal effect is a function of $\bm{\theta}_{11}$,
the estimator for the causal effect is therefore given by
\begin{equation}
\widehat{\zeta}_{j} = \frac{1}{n_{1}}\sum^{n}_{i=1}D_{i}\bm{1}\{Y_{i1} = j\}
- \Big\{
\Phi[(\kappa_{j+1} - \widehat{\mu}_{11}) / \widehat{\sigma}_{11}] -
\Phi[(\kappa_{j} - \widehat{\mu}_{11}) / \widehat{\sigma}_{11}]
\Big\}
\end{equation}
where $n_{1} = \sum^{n}_{i=1}D_{i}$,
and then $\widehat{\Delta}_{j} = \sum^{J-1}_{\ell = j}\widehat{\zeta}_{j}$.
Note that the first term of the right-hand side is a nonparametric estimator of
$\hP(Y_{i1}(1) = j \mid D_{i} = 1)$
 because this quantity is identified from the data without any assumptions.
The second term on the right-hand side is the counterfactual distribution identified by the assumptions (Proposition~\ref{proposition:identification}).

Lemma~\ref{lemma:asymptotic-normality} in Appendix~\ref{appendix:proofs}
establishes the $\sqrt{n}$ consistency of the estimator $\widehat{\bm{\zeta}}$,
whose sampling variance can be derived using the delta method under the independence
assumption.
In practice, however, the block bootstrap can be used to estimate the variance
when outcomes are correlated across time or due to clustering.


\subsection{Assessing the distributional parallel trends assumption}\label{sec:diagnostics}

In the standard DID design, additional pre-treatment periods provide an opportunity
to assess the parallel trends assumption by checking the pre-treatment trends \citep{angrist2008mostly,egami2019how}.
Although it is not a direct test of the assumption,
observing the parallel trends in the pre-treatment periods suggests that
the assumption is more likely to be plausible.
With a similar logic, we can assess the validity of distributional parallel trends assumption (Assumption~\ref{assumption:quantile}).
Specifically, we would expect that if the distributional parallel trends holds for the pre-treatment periods,
it is more reasonable to claim that the assumption holds in the post-period.
Thus, we assess the validity of Assumption~\ref{assumption:quantile} by testing if the similar condition holds for the pre-treatment periods.

\paragraph{The proposed testing procedure}
Suppose that we now observe the outcome for three time periods, $Y_{i0}$, $Y_{i1}$ and $Y_{i2}$
where $Y_{i2}$ is the post-treatment outcome and $Y_{i0}$ and $Y_{i1}$ are the pre-treatment outcomes.
The treatment is administered after time $t = 1$ in this setup,
and thus we have $Y_{it}(0) = Y^{\text{obs}}_{it}$ for $t = 0, 1$ regardless of the treatment status.
This means that observed outcome before the treatment assignment is the same as the potential outcome under the control condition for both treatment and control groups.

Let $\tilde{q}_{d}(v) = \Phi(\mu_{d0}, \sigma_{d0}) \circ \Phi^{-1}(v: \mu_{d1}, \sigma_{d1})$
denote the pre-treatment along of $q_{d}(v)$ defined in Assumption~\ref{assumption:quantile},
where I assume that $U \sim \mathcal{N}(0, 1)$.
Recall that $q_{d}(v)$ captures shift of distributions over time evaluated at quantile $v$. The assumption requires that  two functions are identical on the unit interval, that is, $q_{1}(v) = q_{0}(v)$ for all $v$.
Therefore, we wish to statistically test if $\tilde{q}_{1}(v) = \tilde{q}_{0}(v)$ holds for all $v \in [0, 1]$
using the data from the pre-treatment periods.


Intuitively, we can check the equivalence of two functions $\tilde{q}_{1}$ and $\tilde{q}_{0}$
by assessing the maximum deviation between two functions,
$t_{\max} = \max_{v \in [0,1]} | \tilde{q}_{1}(v) - \tilde{q}_{0}(v)|$.
If this metric is ``small'', we may conclude that $\tilde{q}_{1} = \tilde{q}_{0}$.
Formally, with some threshold $\delta > 0$, we wish to test the following hypotheses:
\begin{align*}
H_{0}\colon \max_{v \in [0,1]} | \tilde{q}_{1}(v) - \tilde{q}_{0}(v)| > \delta
\quad \text{and}\quad
H_{1}\colon  \max_{v \in [0,1]} | \tilde{q}_{1}(v) - \tilde{q}_{0}(v)| \leq \delta
\end{align*}
where $H_{0}$ says that two functions are not equivalent (i.e., large deviation).
Rejecting the null implies that the data supports $H_{1}$ of equivalence
which is what we want to demonstrate.
For now, I assume that researchers know how to choose an appropriate value of $\delta$ based on substantive knowledge.
I will discuss how to calibrate this equivalence threshold in the below.
We can see that the null hypothesis can be written as a union of two hypotheses without absolute values, $H_{0} = H^{+}_{0} \cup H^{-}_{0}$
where
\begin{align*}
H^{+}_{0}\colon \max_{v \in [0,1]} \{ \tilde{q}_{1}(v) - \tilde{q}_{0}(v) \}> \delta
\quad \text{and}\quad
H^{-}_{0}\colon \min_{v \in [0, 1]} \{ \tilde{q}_{1}(v) - \tilde{q}_{0}(v) \} < -\delta.
\end{align*}
This decomposition implies that we can conduct two one-sided tests to determine if we reject the original null $H_{0}$ or not.
In other words, we conclude that $H_{0}$ is false if we reject \textit{both} $H^{+}_{0}$ and $H^{-}_{0}$.

Now, suppose that we construct a $100(1 - \alpha)$\% point-wise confidence interval $[\widehat{L}_{1-\alpha}(v), \widehat{U}_{1-\alpha}(v)]$ for $t(v) \equiv \tilde{q}_{1}(v) - \tilde{q}_{0}(v)$
at each $v$.
The detail of how to construct the confidence interval is presented in Lemma~\ref{lemma:convergence-q1-q0} and \ref{lemma:validity-confidence-set} in Appendix~\ref{appendix:proofs}.
Then, by the one-to-one relationship between the test and the confidence set,
we reject $H^{+}_{0}$ if and only if the upper confidence interval is less than $\delta$, that is
\begin{equation*}
\text{reject } H^{+}_{0} \text{ at $\alpha$ level}\iff
\max_{v \in [0,1]} \widehat{U}_{1-\alpha}(v) < \delta.
\end{equation*}
By the similar argument, we reject $H^{-}_{0}$ at $\alpha$ level if and only if $\min_{v \in [0,1]} \widehat{L}_{1-\alpha}(v) > -\delta$.

%

Proposition~\ref{proposition:validity-test} shows that the proposed procedure is in fact asymptotically level $\alpha$ test, that is, it rejects the null of non-equivalence with probability less than $\alpha$ when the null is true.
\begin{proposition}[Validity of the Testing Procedure]\label{proposition:validity-test}
For a given choice of the equivalence threshold $\delta$ and the level of a test $\alpha$, the testing procedure asymptotically controls the type I error,
that is, under the null $H_{0}\colon t_{\max} \geq \delta$,
\begin{equation*}
\sup_{t\colon \delta \leq |t| < 1}\hP\Big(
\Big\{\max_{v \in [0,1]}\widehat{U}_{1-\alpha}(v) < \delta\Big\} \cap
\Big\{\min_{v \in [0,1]}\widehat{L}_{1-\alpha}(v) \geq - \delta \Big\}
\Big) \leq \alpha
\end{equation*}
as $n \to \infty$.
\end{proposition}
The above proposition shows that
when the equivalence threshold is chosen such that the null is true (i.e., $t_{\max} \geq \delta$),
then the probably to falsely reject the null (type I error) is less than $\alpha$, for any value of $t$
that is consistent with the null.
In other word, the proposed testing procedure is statistically valid
for any choice of the equivalence threshold under the null.

The above result suggests that we can also compute the $p$-value
for this test by solving the rejection rule with respect to $\alpha$,
\begin{equation*}
\widehat{p} = \max\bigg\{\max_{v \in [0,1]}\widehat{p}_{1}(v), \max_{v \in [0,1]}\widehat{p}_{2}(v)\bigg\},
\end{equation*}
where
\begin{equation*}
\widehat{p}_{1}(v)
= 1 - \Phi\Bigg(\frac{\delta - \hat{t}(v)}{\sqrt{\text{Var}(\hat{t}(v)) / n}}\Bigg)
\quad \text{and}\quad
\widehat{p}_{2}(v)
= 1 - \Phi\Bigg(\frac{\delta + \hat{t}(v)}{\sqrt{\text{Var}(\hat{t}(v)) / n}}\Bigg).
\end{equation*}

Intuitively, the \textit{p}-value for the test is the maximum of all point-wise $p$-values
because we are testing the maximum deviation of $\tilde{q}_{1}(v) - \tilde{q}_{0}(v)$.

\paragraph{Choosing an equivalence threshold $\delta$}
So far, we have assumed that researchers have a clear idea what value should be used to assess the equivalence.
When researcher have substantive knowledge about the appropriate value of $\delta$ given an application,
it is  reasonable to choose $\delta$ according to the knowledge.
Oftentimes, this approach might not be feasible since it is not straightforward to form an idea of what value of $\delta$ should be deemed appropriate,
especially when the value of $\delta$ is not directly tied to interpretable quantities such as causal effects.
Although, any choice of $\delta$ is a valid choice because type I error is controlled for the corresponding null hypothesis,
it seems useful to suggest a reasonable default value of $\delta$ to facilitate the practical use of the  method.

I suggest the following value of $\delta$ as a reasonable starting point,
\begin{equation*}
\delta_{n} = \min\bigg\{
1.2\sqrt{\frac{n_{1} + n_{0}}{n_{1}n_{0}}},\ 1
\bigg\}
\end{equation*}
where $1.2 \approx \sqrt{-\log(\omega)/2}$ with $\omega = 0.05$
and $\sqrt{(n_{1} + n_{0}) / (n_{1}n_{0})} \sim n^{-1/2}$ when $n_{1} \sim n_{0}$.
This is a threshold used in the conventional KS test which is a nonparametric test
on the difference between two distribution functions.
The test is based on the maximum difference between two cumulative distribution functions.
In KS test, the value of $\omega$ by the level of a test,
but it is fixed here.
The key feature of this threshold is that
$\delta_{n}$ depends on the sample size.
The equivalence threshold that depends on the sample size is discussed in \cite{romano2005optimal}.
Intuitively, this selection of $\delta$ implies that we raise the standard of what the equivalence means
as the sample size increases.
Therefore, rejecting the null with larger $n$ will be a stronger evidence for the identification assumption.

\section{Empirical Findings}\label{sec:application}

In this section, I revisit the empirical application introduced in Section~\ref{sec:application-introduction}.
We first reanalyze the two-wave panel of CCES (2010--2012).
This two-wave panel is the data used for main analyses in the original studies.
We then analyze the three-wave panel  of CCES (2010--2012--2014) which allows us to assess the identification assumption using the pre-treatment periods.

In the following, we focus on estimating the following causal quantities:
\begin{align*}
\zeta_{j} & = \hP(Y_{i1}(1) = j \mid D_{i} = 1) -
               \hP(Y_{i1}(0) = j \mid D_{i} = 1)
\end{align*}
for $j = 0, 1, 2$.
Recall that \texttt{less-strict} is coded as \texttt{0},
and  \texttt{more-strict} category is coded as \texttt{2}.


\subsection{Result from the two-wave panel}
This section presents a result of the analysis on the two-wave sample from CCES ($n = 16620$).
The outcome is measured in 2010 and 2012 and I treat a response in 2012 as the post-treatment outcome.
Respondents living in a neighborhood where mass shootings happened within 100 miles between 2010 and 2012 are considered as treated $(n_{1} = 4893)$.
In total, there were 16 mass shooting incidents recorded in the dataset between the two waves of CCES \citep[][Appendix C]{newman2019mass}.
In addition to the analysis with the full sample,
I also investigate effect heterogeneity by pre-treatment covariates.
First, I investigate if the baseline safety of the neighborhood affects how people respond to mass shootings.
Respondents are classified into either ``prior exposure'' group or ``no prior exposure'' group.
A respondent is in the ``no prior exposure'' group if she is living in a neighborhood that did not have
mass shootings within 100 miles of the area for the last ten years (as of 2010).
We would expect that people react differently to mass shootings depending on how frequent these events are in their life.
Second, following the original papers, I investigate if effects vary across respondents' party affiliations.
Since issues related to gun control are debated along the party line in the US, we might expect that
people react differently depending on which party they affiliate with.

Figure~\ref{fig:estimated-results-twowaves} shows the results.
In the figure, circles represent point estimates for $\zeta_{0}$
which can be interpreted as the causal effect on preferring less strict gun regulations,
while triangles shows estimates for $\zeta_{2}$ which captures the effect on preferring more strict control of firearm sales; squares are estimate for the middle category ($\zeta_{1}$), which can be interpreted as a preference to the status quo.
\begin{figure}[htb]
  \centerline{\includegraphics[scale=0.8]{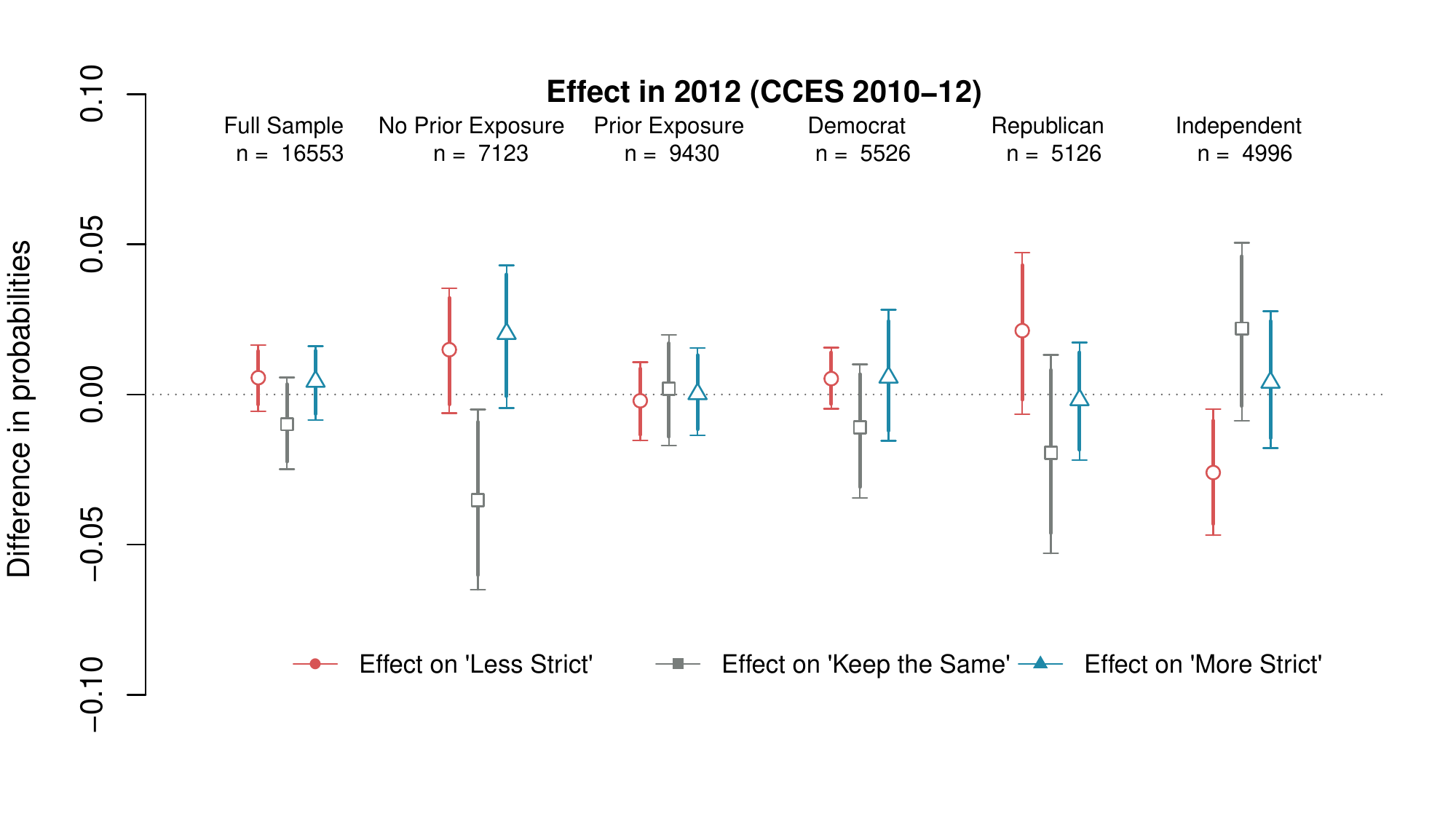}}
  \caption{Estimated treatment effects with 90\% (solid) and 95\% (thick) confidence intervals. Circles indicate effect for \texttt{less-strict} ($\zeta_{0}$),
  squares for \texttt{keep-the-same} ($\zeta_{1}$) and triangles for \texttt{more-strict} ($\zeta_{2}$). Labels above estimates indicate subsamples used for the analysis where $n$ indicates the size of the sample.}
  \label{fig:estimated-results-twowaves}
\end{figure}
Along with point estimates, I also show the uncertain estimates.
Thick (thin) vertical lines show 90\% (95\%) confidence intervals calculated via block bootstraps. To account for the fact that the treatment assignment is at the zip code level, the bootstrap is conducted blocking at the zip code level. There are $9042$ unique zip codes in the two-wave sample. I sample zip codes with replacement and create bootstrap samples. Confidence intervals are based on $2000$ bootstrap iterations.
Text labels shown above estimates indicate the subsamples  and their sample sizes used for the analysis.

We can see that causal effect estimates are not statistically significant at the 10\% level for all categories in the full sample.
Estimates are precisely estimated and they are all close to zero, indicating that there is little evidence to suggest that mass shootings have, on average, any effect on the attitude towards gun control regulations among those who live in their vicinity of public shootings.

Following the original authors, I conduct two sets of subgroup analysis by ``prior exposure'' status
and  by the partisanship.
The analysis reveals a similar pattern that most of the estimates are not statistically distinguishable from zero at the conventional level.
However, we can also see that heterogeneity exists:
the ``no prior exposure'' group has negative effect for the middle category ($\widehat{\zeta}_{1} = -0.035$, $\text{SE} = 0.016$) which is statistically significant at the 5\% level.
This result implies that those living in the safer neighborhood (i.e., ``no prior exposure'')
move away from the status quo.
Although not statistically significant at the 10\% level, we can also see that
the effect on the less strict category is positively estimated  for this ``no prior exposure'' group, indicating that the shift away from the status quo was probably not uni-directional.
We can also see the negative effect on \texttt{less-strict} category among the independents ($\widehat{\zeta}_{0} = -0.026$, $\text{SE} = 0.011$), which is statistically distinguishable from zero at the 5\% level.
In Appendix~\ref{appendix:subset}, I present a result using a more granular measure of partisanship used in the survey,
which asks respondents to categorize themselves on a 7-point scale from ``strong Democrat'' to ``Strong Republican.'' The result in Figure~\ref{fig:twowaves-full-pid7} shows that the effect is concentrated among lean Democrats, 91\% of them  categorized themselves as ``independent'' on the 3-point partisanship scale.

%
\begin{figure}[htb]
  \centerline{\includegraphics[scale=0.8]{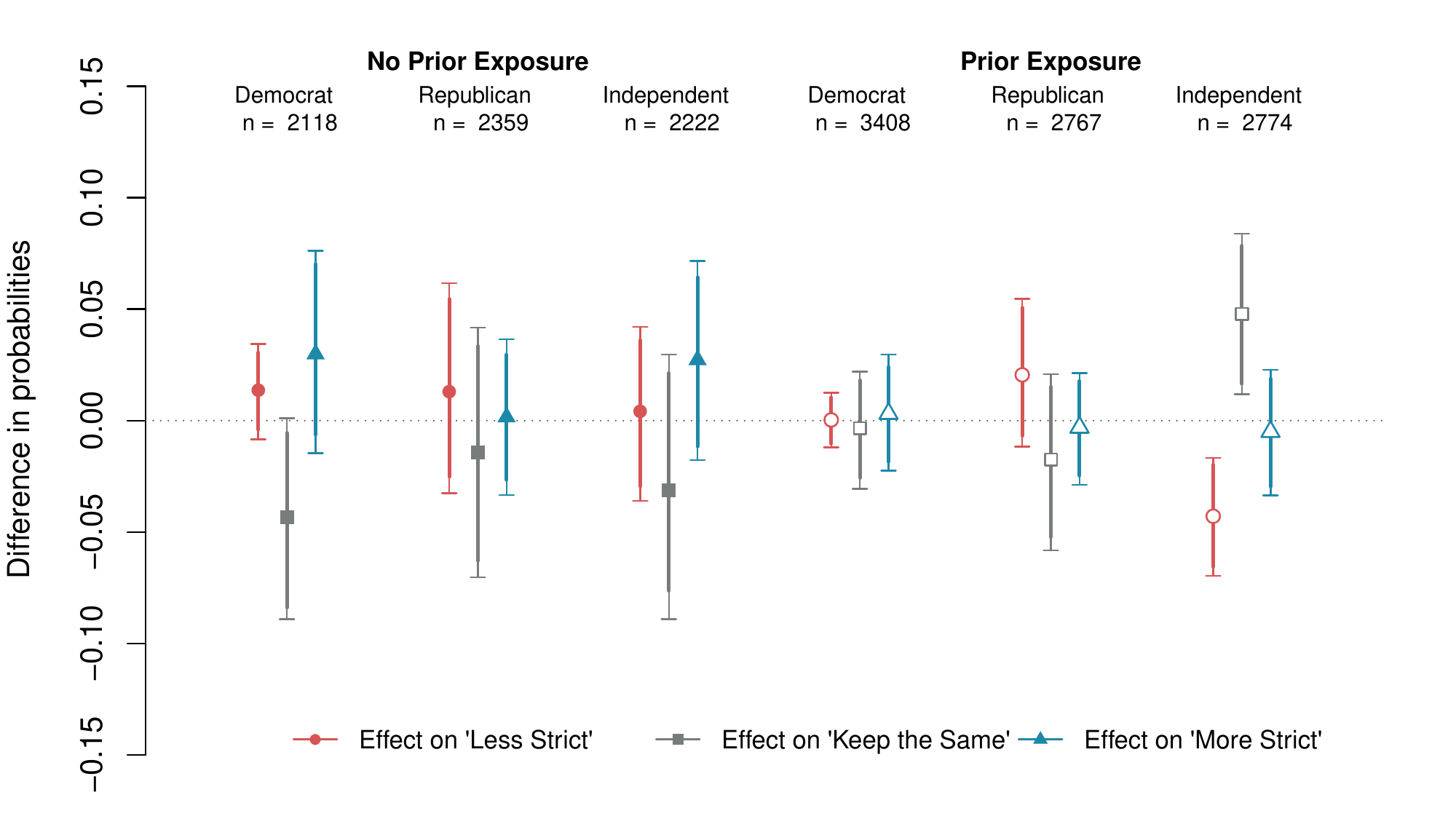}}
  \caption{Estimated treatment effects with 90\% (solid) and 95\% (thick) confidence intervals. Circles indicate effects for \texttt{less-strict} ($\zeta_{0}$), squares for \texttt{keep-as-they-are} ($\zeta_{1}$) and triangles  for \texttt{more-strict} ($\zeta_{2}$). Text labels above the estimates indicate subsamples used for the analysis.}
  \label{fig:estimated-results-twowaves-interactions}
\end{figure}

To further investigate the interactive effects between the the partisanship and prior exposure status,
I considered interactions between the two variables.
Figure~\ref{fig:estimated-results-twowaves-interactions} shows the results of the analysis.
As we can see the effect is concentrated among Democrats who are in the ``no prior exposure'' group,
while none of the effects are statistically significant for other partisans.
The figure also shows that partisanship does not play a role in the ``prior exposure'' group
where estimates are indistinguishable from zero.

Finally, \cite{barney2019reexamining}
consider different thresholds to determine who are ``exposed'' to the mass shootings.
In addition to the above 100 mile threshold, I also estimate effects for the 25 mile threshold.
The result is shown in Figure~\ref{fig:twowave-main-25mi} in Appendix~\ref{appendix:subset}.
We observe similar patterns with the previous results,
while there are two notable differences. First, $\zeta_{1}$ is now statistically significant at the 10\% level ($\zeta_{1} =  -0.021$, $\text{SE} =   0.012$).
Second,  effects are  clearer for Democrats without prior exposure: $\zeta_{1} < 0$ and $\zeta_{2} > 0$ and both of the estimates are statistically significant at the 5\% level.

\subsection{Diagnostics using three-wave panel}

Next, I analyze the three-wave panel from CCES (2010-12-14) to assess if the identification assumption
made in Assumption~\ref{assumption:quantile} is plausible or not.
I subset the dataset so that I include only two types of respondents: those who experienced the mass shootings
only after 2012 (treated group) and those who never experience the mass shootings throughout the sample periods (control group).
This allows us to treat 2010 and 2012 as the pre-treatment periods,
because no one in this subsample is affected by the treatment happened before 2012.
To avoid the possibility that the past exposure might affect the baseline attitudes,
I further condition on the prior-exposure variable, including only respondents who are in the ``no prior exposure'' group.
This subset consists of $2817$ respondents among which $667$ respondents are eventually treated between 2012 and 2014.
In total, there were 28 incidents of mass shootings recorded in the dataset that happened between 2012 and 2014 waves.

I apply the diagnostic test proposed in Section~\ref{sec:diagnostics}
to the pre-treatment outcome.
The goal here is to statistically test if the condition of the distributional parallel trend holds, namely,
$\tilde{q}_{1}(v) = \tilde{q}_{0}(v)$ where $\tilde{q}_{d}(v)$ is the pre-treatment analog of the quantile-quantile relationship defined on group $d$ (i.e., $q_{d}(v)$ in Assumption~\ref{assumption:quantile}).
Specifically, I test the null hypothesis of non-equivalence, $H_{0}\colon \tilde{q}_{1}(v) \neq \tilde{q}_{0}(v)$ for all $v$ against the equivalence.

I compute the test statistic $\hat{t}_{\max} = \max_{v} \hat{t}(v)$, where $\hat{t}(v) = \widehat{\tilde{q}}_{1}(v) - \widehat{\tilde{q}}_{0}(v)$, and corresponding confidence intervals at the 5\% level.
Each $\hat{t}(v)$ is computed by evaluating $\widehat{\tilde{q}}_{1}$ and $\widehat{\tilde{q}}_{0}$ on the finite number of grid points between $0.001$ and $0.999$  where the distance between points is set to $0.01$.
The equivalence threshold is chosen based on the heuristic criterion discussed in Section~\ref{sec:diagnostics},
$\delta_{n} = \sqrt{-\log(0.05)/2 \times n /(n_{1}n_{0})} \approx 0.054 $
where $n = 2817$ and $n_{1} = 667$.

\begin{figure}[htb]
  \centerline{\includegraphics[scale=0.65]{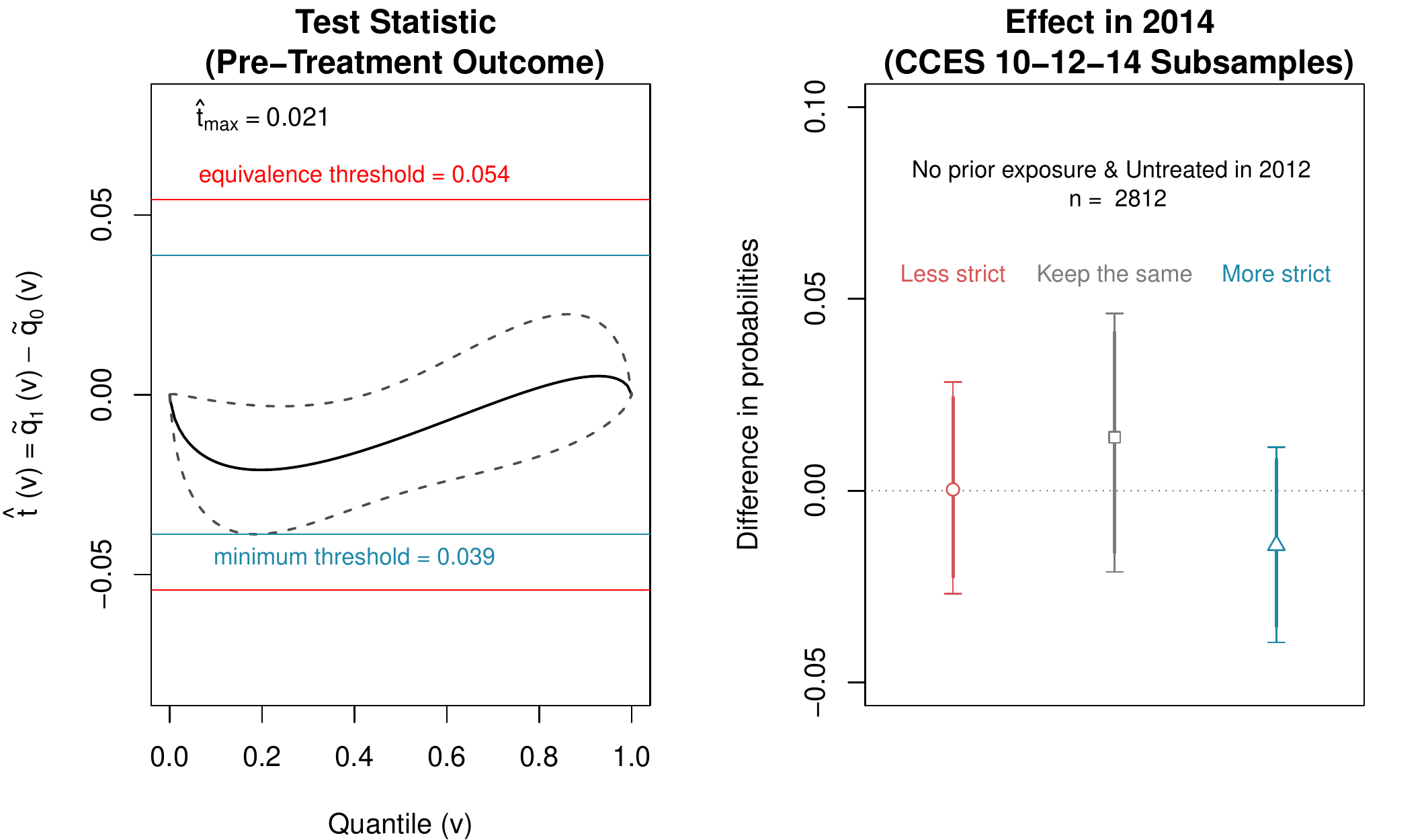}}
    \caption{\textbf{Left} -- Test statistics $\hat{t}(v)$ (solid line) with point-wise 95\% confidence intervals (dashed lines).
    \textcolor{red}{Red} lines show the equivalence range $[-\delta_{n}, \delta_{n}]$.
    The figure shows that the largest (smallest) point of the upper (lower) confidence intervals is strictly contained in the equivalence range. It suggests that the null is rejected at the 5\% level with $\delta_{n} = 0.054$.
    \textbf{Right} -- Estimated causal effects with 90\% (thick) and 95\% (thin) confidence intervals.
    Either effects are not statistically distinguishable from zero at the 10\% level.}
    \label{fig:estimated-results-threewaves}
\end{figure}

The left panel of Figure~\ref{fig:estimated-results-threewaves}
shows the difference between the two functions, $\hat{t}(v)$, evaluated at a value $v$ on the unit interval (solid line).
Dashed lines show the point-wise 95\% confidence intervals.
The test statistic (the estimated largest deviation) is $\hat{t}_{\max} = 0.021$
with the largest upper bound $\widehat{U}_{\max} = 0.022$ and the smallest lower bound $\widehat{L}_{\min} = -0.039$.
Since the confidence range $[\widehat{L}_{\min}, \widehat{U}_{\max}]$ is strictly contained in
the equivalence range $[-\delta_{n}, \delta_{n}]$, we can reject the null of non-equivalence at the 5\% level ($p = 0.001$).
In other words, for any choice of $\delta$ that is greater than $\max\{|\widehat{U}_{\max}|, |\widehat{L}_{\min}|\} = 0.039$, we reject the null  at the 5\% level.
The result suggests that during the pre-treatment periods the data supports the analogous condition of Assumption~\ref{assumption:quantile}.

After confirming the plausibility of the key identification assumption,
we now analyze the outcome measured in 2012 (pre-treatment) and 2014 (post-treatment) to estimate the causal effect
for the three-wave subsample.
The right panel of Figure~\ref{fig:estimated-results-threewaves} shows the result of the analysis.
We can see that none of the estimates are  statistically distinguishable from zero at the 10\% level ($\widehat{\zeta}_{0} = 0.000$, $\text{SE} = 0.0144$; $\widehat{\zeta}_{1} = 0.0140$, $\text{SE} = 0.0172$; and $\widehat{\zeta}_{2} = -0.0142$, $\text{SE} = 0.0132$).
This result somewhat contradicts findings in the previous section,
where I found a negative effect on $\zeta_{1}$ among the ``no prior exposure'' group.
There are many possible reasons why effects could vary over time.
One possibility is simply the size of the dataset.
The subset of the three-wave panel has smaller respondents than the two-wave samples analyzed in the previous section. This difference obviously translates into differences in uncertainty estimates.
Another possibility is due to the contextual differences.
On December 14th, after the 2012 wave of the CCES, the Sandy Hook Elementary School shooting occured.
This was one of the deadliest mass shootings in the US history,
which possibly raised the salience of the issue affecting gun control regulations nationally.

\section{Concluding Remarks}\label{sec:conclusion}
In spite of the recent developments in the literature on the DID design,
less attention has been paid for when the outcome is measured on an ordinal scale.
In this paper, I proposed a method that allows scholars to leverage ordinal outcomes
without making  the linearity assumption as in the the standard DID analysis.
I also proposed a procedure that assesses if the key identification assumption is plausible
when additional pre-treatment periods are available.
This enable scholars to inspect the data and to discuss if the assumption is reasonable
given a particular dataset they analyze,
which is a crucial step for any research that attempts to establish a causal relationship.

Several extensions of the proposed methods are possible.
In Appendix~\ref{appendix:extensions}, I demonstrate that
the proposed method can be useful to estimate other types of causal estimands
such as the proportion of who benefits from the treatment \citep[e.g.][]{lu2018treatment}.
Recent years, it has been argued that such estimands are preferable because
the distributional treatment effects considered in the main text are not necessarily easy to interpret.
Because the proposed method identifies the entire distributions of the potential outcomes,
it is possible to compute any causal estimand that is a function of marginal distributions of the potential outcome.
In the appendix, I also discuss how to incorporate time-varying covariates,
which requires a further modeling assumption.

In Appendix~\ref{sec:simulation},
I present a simulation study where I assess the finite sample property of the proposed estimator
and the testing procedure.
I find that under the correct model specification, the proposed method outperforms
the two competing methods: the standard difference-in-differences on the dichotomized outcome
and the ordered probit model.
I also find that the type-I error is properly controlled for the proposed testing procedure, while the power depends on the choice of the equivalence threshold.

Finally, future research should consider an extension to a complex design.
Specifically, the staggered adoption design where the treatment is assigned over time
is a popular data structure in applied studies.
Although the development of such methods is beyond the scope of this paper,
the same framework should improve the analysis of ordinal outcome beyond the standard DID setting.

%
%

\clearpage
\bibliography{ord_did.bib}

\clearpage
\appendix
\renewcommand\thefigure{\thesection.\arabic{figure}}
\setcounter{figure}{0}
\setstretch{1.15}

\begin{center}
 {\LARGE \textbf{Appendix}}
\end{center}
\thispagestyle{empty}


\section{Proofs of Propositions}\label{appendix:proofs}

\subsection{Lemmas}

Before proving propositions, we present useful lemmas.

\begin{lemma}[Identification of mean and variance of the latent variables.]\label{lemma:latent-identification}
  Suppose that the cutoffs are fixed at $\kappa_{1}$ and $\kappa_{2}$ for $Y_{dt} = j \in \{0, 1, 2\}$.
  Then, $\mu_{dt}$ and $\sigma_{dt}$ in  $Y^{*}_{dt} \sim \mu_{dt} + \sigma_{dt}U$ are uniquely identified
  from the observed probability distribution.
\end{lemma}

\begin{proof}[Proof of Lemma~\ref{lemma:latent-identification}]
  Suppose that $U$ has the density $f_{U}(u)$.
  Then, we can form a non-linear system of equations
  \begin{align*}
  \Pr(Y_{dt} = 0) &= \int^{\kappa_{1}}_{-\infty} f_{U}((y^{*} - \mu_{dt}) / \sigma_{dt})dy^{*}\\
  \Pr(Y_{dt} = 2) &= \int^{\infty}_{\kappa_{2}}f_{U}((y^{*} - \mu_{dt}) / \sigma_{dt})dy^{*}
  \end{align*}
  which are sufficient for estimating $\mu$ and $\sigma$.
\end{proof}

\begin{lemma}[Alternative formula for identification]\label{lemma:alternative-identification}
Suppose $Y_{dt} = j \in \{0, 1, 2\}$.
Let $v_{1} = F_{01}(\kappa_{1})$ and $v_{2} = F_{01}(\kappa_{2})$
where $\bm{\kappa}$ is a set of fixed cutoffs.
Under Assumption~\ref{assumption:model}, \ref{assumption:distribution} and \ref{assumption:quantile},
we identify $\mu_{11}$ and $\sigma_{11}$ by the following system of non-linear equations:
\begin{align*}
q_{0}(v_{1}) &= \int^{F^{-1}_{10}(v_{1})}_{-\infty} f_{U}((y^{*} - \mu_{11}) / \sigma_{11})dy^{*}\\
q_{0}(v_{2}) &= \int^{F^{-1}_{10}(v_{2})}_{-\infty} f_{U}((y^{*} - \mu_{11}) / \sigma_{11})dy^{*}.
\end{align*}
\end{lemma}

\begin{proof}[Proof of Lemma~\ref{lemma:alternative-identification}]

Under the distributional parallel trends assumption, we have $q_{0}(v) = q_{1}(v)$ for all $v \in [0, 1]$.
Then,
\begin{align*}
q_{0}(v)
&= F_{11} \circ F^{-1}_{10}(v)\\
&= \int^{F^{-1}_{10}(v)}_{-\infty} f_{11}(y^{*})dy^{*}\\
&= \int^{F^{-1}_{10}(v)}_{-\infty} f_{U}((y^{*} - \mu_{11})/\sigma_{11})dy^{*}
\end{align*}
where the first equality is due to the definition of $q_{d}(v)$
and the last equality follows by Assumption~\ref{assumption:distribution}.
Pick $v_{1}$ and $v_{2}$ as in the statement.
Drawing on a similar to the argument in Lemma~\ref{lemma:latent-identification}, we obtain the identification.

\end{proof}

\begin{lemma}[Invariance of $\widehat{\zeta}_{j}$ under different cutoffs]\label{lemma:invariance-cutoff}
Suppose we have two sets of cutoffs $\bm{\kappa}$ and $\bm{\kappa}'$ ($\bm{\kappa} \neq \bm{\kappa}'$)
for $Y_{dt} = j \in \{0, 1, 2\}$.
Then, $\widehat{\bm{\zeta}}(\bm{\kappa}) = \widehat{\bm{\zeta}}(\bm{\kappa}')$.
\end{lemma}

\begin{proof}[Proof of Lemma~\ref{lemma:invariance-cutoff}]
Let $F_{dt}(y) = \Pr(Y^{*}_{dt} \leq y)$ denote the cumulative distribution function of the latent variable $Y^{*}_{dt}$ under the cutoff $\bm{\kappa}$.
Similarly, let $\widetilde{F}_{dt}(y)$ denote the CDF under $\bm{\kappa}'$.
To show the causal effect estimates are invariant to the choice of cutoffs, it is sufficient to show that $F_{11}(y) = \widetilde{F}_{11}(y)$, that is, the invariance of the identified counterfactual latent distribution.

We first show that $F_{00}(F^{-1}_{01}(u)) = \widetilde{F}_{00}(\widetilde{F}^{-1}_{01}(u'))$
for $u = F_{01}(\kappa_{1})$ and $u' = \widetilde{F}_{01}(\kappa'_{1})$.
Now, note that we have $u = u'$ because
\begin{align*}
F_{01}(\kappa_{1})
  &= \Pr(Y_{01} = 0) \\
  &= \widetilde{F}_{01}(\kappa'_{1})
\end{align*}
where $Y_{01}$ is the observed outcome.
Thus,
\begin{align*}
F_{00}(F^{-1}_{01}(u)) - u
&= F_{00}(\kappa_{1}) - F_{01}(\kappa_{1}) \\
&= \Pr(Y_{00} = 0) - \Pr(Y_{01} = 0)\\
&= \widetilde{F}_{00}(\kappa_{1}) - \widetilde{F}_{01}(\kappa_{1})\\
&= \widetilde{F}_{00}(\widetilde{F}^{-1}_{01}(u')) - u'
\end{align*}
which proves that $F_{00}(F^{-1}_{01}(u)) = \widetilde{F}_{00}(\widetilde{F}^{-1}_{01}(u'))$.

Next, by the similar argument, we have that $F_{10}(\kappa_{1}) = \widetilde{F}_{10}(\kappa'_{1})$, because
\begin{align*}
F_{10}(\kappa_{1})
&= \Pr(Y_{10} = 0)\\
&= \widetilde{F}_{10}(\kappa'_{1}).
\end{align*}

Repeating the above two steps for $\kappa_{2}$ and $\kappa'_{2}$,
we obtain that
\begin{align*}
F_{00}(F^{-1}_{01}(u))
&= \int^{F^{-1}_{10}(u)}_{-\infty} f_{U}((y^{*} - \mu_{11}) / \sigma_{11})dy^{*}\\
&= \widetilde{F}_{00}(\widetilde{F}^{-1}_{01}(u'))
\end{align*}
and
\begin{align*}
F_{00}(F^{-1}_{01}(v))
&= \int^{F^{-1}_{10}(v)}_{-\infty} f_{U}((y^{*} - \mu_{11}) / \sigma_{11})dy^{*}\\
&= \widetilde{F}_{00}(\widetilde{F}^{-1}_{01}(v'))
\end{align*}
where $v = F_{01}(\kappa_{2})$ and $v' = \widetilde{F}_{01}(\kappa'_{2})$.

Applying the result of Lemma~\ref{lemma:alternative-identification},
we can see that $\mu_{11}$ and $\sigma_{11}$ are uniquely identified under different sets
of cutoffs, that is, $\bm{\zeta}(\bm{\kappa}) = \bm{\zeta}(\bm{\kappa}')$.

Finally, replacing all quantities with their sample analog, we conclude that
$\widehat{\bm{\zeta}}(\bm{\kappa}) = \widehat{\bm{\zeta}}(\bm{\kappa}')$.
\end{proof}

\begin{lemma}[Asymptotic Normality of Causal Estimates]\label{lemma:asymptotic-normality}
Under some regularity conditions, as $n \to \infty$ with $n_{1} / n \to k$,
we have that
\begin{equation}
\sqrt{n}(\widehat{\zeta}_{j} - \zeta_{j})  \leadsto \mathcal{N}(0, \sigma^{2}_{j})
\end{equation}
\end{lemma}

\begin{proof}[Proof of Lemma~\ref{lemma:asymptotic-normality}]

We prove the statement by showing the following two statements:
\begin{enumerate}
  \item $\sum^{n}_{i=1}D_{i}\bm{1}\{Y_{i1} = j\} / n_{1}$ is $\sqrt{n}$-consistent estimator for $\Pr(Y_{i1}(1) = j \mid D_{i} = 1)$.
  \item $\widehat{\bm{\theta}}_{11} = (\widehat{\mu}_{11}, \widehat{\sigma}_{11})^{\top}$ is $\sqrt{n}$-consistent estimator for $\bm{\theta}_{11}$.
\end{enumerate}
We then use the continuous mapping theorem for the convergence in distribution
to obtain the final result.

To be clear, we (sometime implicitly) condition on $D_{i} = 1$ throughout the proof,
which assumes that
there is a super population of units with $D_{i} =1$.
Now let $W_{i} = D_{i} \bm{1}\{Y_{i1} = j\}$
and $\pi_{11}(j) = \Pr(Y_{i1}(1) = j \mid D_{i} = 1)$.
Under the assumption that $Y_{it} \indep Y_{i't'}$ for any combination of $i$ and $t$, it follows that
$\sum^{n}_{i=1}W_{i} / n_{1} \to \E[\bm{1}\{Y_{i1}(1) = j \} \mid D_{i} = 1] = \pi_{11}(j)$ as $n \to \infty$ with $n_{1} / n \to k$ by the law of large numbers.
This proves the consistency.
By the central limit theorem, it also follows that
\begin{align*}
\sqrt{n}\bigg(\frac{1}{n_{1}}\sum^{n}_{i=1}W_{i} - \pi_{11}(j)\bigg)
&= \frac{n}{n_{1}}\frac{1}{\sqrt{n}}\sum^{n}_{i=1}(W_{i} - \pi_{11}(j))\\
&\leadsto \mathcal{N}\bigg(0, \frac{\text{Var}(\widehat{\pi}_{11}(j))}{k^{2}}\bigg)
\end{align*}
as $n \to \infty$ with $n_{1}/n \to k$.

Next, recall that $\widehat{\bm{\theta}}_{11}$ is given as a transformation of $\widehat{\bm{\theta}}_{00}$, $\widehat{\bm{\theta}}_{01}$ and $\widehat{\bm{\theta}}_{10}$,
all of which are MLE of the problem described in Section~\ref{subsec:estimation}.
Therefore, under the assumption of correct model specification, we obtain that
$\widehat{\bm{\theta}}_{00}$, $\widehat{\bm{\theta}}_{01}$ and $\widehat{\bm{\theta}}_{10}$ are jontly asymptotically normal centered around $\bm{\theta}_{00}$, $\bm{\theta}_{01}$ and $\bm{\theta}_{10}$.
By the continuous mapping theorem, it follows that $\widehat{\bm{\theta}}_{11}$ is also asymptotically normally distributed.

Finally, using the continuous mapping theorem again, we have that $\widehat{\zeta}_{j}$
is a $\sqrt{n}$ consistent estimator for $\zeta_{j}$
where the variance $\sigma_{j}$ is obtained by the delta method (See Lemma~\ref{lemma:convergence-q1-q0}).

\end{proof}

\begin{lemma}[Asymptotic Normality of $\bm{\theta}$ for Pre-treatment Parameters]\label{lemma:asymptotic-normality-theta}
Let $\bm{\theta} = (\bm{\theta}^{\top}_{00}, \bm{\theta}^{\top}_{01}, \bm{\theta}^{\top}_{10}, \bm{\theta}^{\top}_{11})^{\top}$,
all of which are estimated using  the data from the pre-treatment periods.
Then, under regularity conditions in \cite{newey1994large},
the maximum likelihood estimator $\widehat{\bm{\theta}}$ is asymptotically normal with covariance $\Omega$,
\begin{equation}
\sqrt{n}(\widehat{\bm{\theta}} - \bm{\theta}) \leadsto \mathcal{N}(0, \Omega)
\end{equation}
where $\Omega$ is a block-diagonal matrix under independence assumption.
\end{lemma}

\begin{proof}[Proof of Lemma~\ref{lemma:asymptotic-normality-theta}]
The result is a direct application of the standard result of the maximum likelihood estimation.
Therefore, proof is omitted.
\end{proof}

\begin{lemma}[Asymptotic Distribution of the Test Statistic]\label{lemma:convergence-q1-q0}
Assume that $U \sim \mathcal{N}(0, 1)$.
Let $t(v; \bm{\theta}) = \tilde{q}_{1}(v; \bm{\theta}) - \tilde{q}_{0}(v; \bm{\theta})$
and $\hat{t}(v) \equiv t(v; \widehat{\bm{\theta}})$.
Then, we have that
\begin{equation}
\sqrt{n}(t(v; \widehat{\bm{\theta}}) - t(v; \bm{\theta}))  \leadsto \mathcal{N}(0, \mathrm{Var}(\hat{t}(v)))
\end{equation}
for each $v \in [0, 1]$ with
\begin{equation}
\mathrm{Var}(\hat{t}(v)) =
\bigg(\frac{\partial}{\partial \bm{\theta}}t(v; \bm{\theta})\bigg)^{\top}\Omega
\bigg(\frac{\partial}{\partial \bm{\theta}}t(v; \bm{\theta})\bigg)
\end{equation}
where $\bm{\theta}$ is evaluated at the truth, $\Omega$ is the asymptotic variance covariance matrix of $\widehat{\bm{\theta}}$ given in Lemma~\ref{lemma:asymptotic-normality-theta},
and the gradient takes the form of
\begin{align*}
\frac{\partial}{\partial \bm{\theta}}t(v; \bm{\theta})
=
\left[
\begin{array}{c}
\exp(-z^{2}_{0})/\sqrt{2\pi}\sigma_{00}\\
\exp(-z^{2}_{0})z_{0} / \sqrt{\pi}\sigma_{00}\\
-\exp(-z^{2}_{0})/\sqrt{2\pi}\sigma_{00}\\
-\exp(-z^{2}_{0})\mathrm{erf}^{-1}(2v-1) / \sqrt{\pi}\sigma_{00}\\
-\exp(-z^{2}_{1})/\sqrt{2\pi}\sigma_{10}\\
- \exp(-z^{2}_{1}) z_{1} / \sqrt{\pi}\sigma_{10}\\
\exp(-z^{2}_{1})/\sqrt{2\pi}\sigma_{10}\\
\exp(-z^{2}_{1})\mathrm{erf}^{-1}(2v-1)/ \sqrt{\pi}\sigma_{10}
\end{array}
\right]
\end{align*}
with
\begin{align*}
z_{d} \equiv \frac{\mu_{d1} - \mu_{d0}}{\sigma_{d0}\sqrt{2}} + \frac{\mathrm{erf}^{-1}(2v -1)}{\sigma_{d0}/ \sigma_{d1}}.
\end{align*}

\end{lemma}

\begin{proof}[Proof of Lemma~\ref{lemma:convergence-q1-q0}]
Recall that the derivative of the error function is given by
\begin{equation}
\frac{d}{dz}\mathrm{erf}(z) = \frac{2}{\sqrt{\pi}}e^{-z^{2}}
\end{equation}
which is differentiable with respect to $z$.

Now, I compute the derivative of $t(v; \bm{\theta})$ with respect to $\bm{\theta}$,
\begin{align*}
\frac{\partial}{\partial \bm{\theta}}t(v; \bm{\theta})
=
\left[
\begin{array}{c}
\bigg(\frac{\partial}{\partial \bm{\theta}_{00}}t(v; \bm{\theta})\bigg)\\
\bigg(\frac{\partial}{\partial \bm{\theta}_{01}}t(v; \bm{\theta})\bigg)\\
\bigg(\frac{\partial}{\partial \bm{\theta}_{10}}t(v; \bm{\theta})\bigg)\\
\bigg(\frac{\partial}{\partial \bm{\theta}_{11}}t(v; \bm{\theta})\bigg)\\
\end{array}
\right]
=
\left[
\begin{array}{c}
\exp(-z^{2}_{0})/\sqrt{2\pi}\sigma_{00}\\
\exp(-z^{2}_{0})z_{0} / \sqrt{\pi}\sigma_{00}\\
-\exp(-z^{2}_{0})/\sqrt{2\pi}\sigma_{00}\\
-\exp(-z^{2}_{0})\mathrm{erf}^{-1}(2v-1) / \sqrt{\pi}\sigma_{00}\\
-\exp(-z^{2}_{1})/\sqrt{2\pi}\sigma_{10}\\
- \exp(-z^{2}_{1}) z_{1} / \sqrt{\pi}\sigma_{10}\\
\exp(-z^{2}_{1})/\sqrt{2\pi}\sigma_{10}\\
\exp(-z^{2}_{1})\mathrm{erf}^{-1}(2v-1)/ \sqrt{\pi}\sigma_{10}
\end{array}
\right]
\end{align*}
where
\begin{align*}
z_{d} \equiv \frac{\mu_{d1} - \mu_{d0}}{\sigma_{d0}\sqrt{2}} + \frac{\mathrm{erf}^{-1}(2v -1)}{\sigma_{d0}/ \sigma_{d1}}
\end{align*}

Given that Lemma~\ref{lemma:asymptotic-normality-theta} establishes the asymptotic normality of $\widehat{\bm{\theta}}$,
the result immediately follows by the application of the Delta method.
Then, we get
\begin{equation}
\sqrt{n}(t(v; \widehat{\bm{\theta}}) - t(v; \bm{\theta})) \leadsto
\bigg(\frac{\partial}{\partial \bm{\theta}}t(v; \bm{\theta})\bigg)
\mathcal{N}(0, \Omega)
\end{equation}

From here, we obtain the variance formula as
\begin{align*}
\mathrm{Var}(\hat{t}(v))
&=
\bigg(\frac{\partial}{\partial \bm{\theta}_{00}}t(v)\bigg)^{\top}\Omega_{00}\bigg(\frac{\partial}{\partial \bm{\theta}_{00}}t(v)\bigg) +
\bigg(\frac{\partial}{\partial \bm{\theta}_{01}}t(v)\bigg)^{\top}\Omega_{01}\bigg(\frac{\partial}{\partial \bm{\theta}_{01}}t(v)\bigg) \\
&\quad +
\bigg(\frac{\partial}{\partial \bm{\theta}_{10}}t(v)\bigg)^{\top}\Omega_{10}\bigg(\frac{\partial}{\partial \bm{\theta}_{10}}t(v)\bigg) +
\bigg(\frac{\partial}{\partial \bm{\theta}_{11}}t(v)\bigg)^{\top}\Omega_{11}\bigg(\frac{\partial}{\partial \bm{\theta}_{11}}t(v)\bigg)
\end{align*}
where $\hat{t}(v) \equiv t(v; \widehat{\theta})$.
\end{proof}

\begin{lemma}[Validity of $(1-\alpha$) level sets \citep{liu2009assessing}]\label{lemma:validity-confidence-set}

Let $t(v) = \tilde{q}_{1}(v) - \tilde{q}_{0}(v)$.
Suppose $\widehat{U}_{1-\alpha}(v)$ and $\widehat{L}_{1-\alpha}(v)$
are point-wise upper and lower $(1-\alpha)$ level confidence intervals, respectively
such that $\widehat{U}_{1-\alpha}(v) = \hat{t}(v) + \Phi^{-1}(1-\alpha)\sqrt{\mathrm{Var}(\hat{t}(v))/n}$
and $\widehat{L}_{1-\alpha}(v) = \hat{t}(v) - \Phi^{-1}(1-\alpha)\sqrt{\mathrm{Var}(\hat{t}(v))/n}$.
Then,
\begin{align}
\hP\bigg(\max_{v\in [0,1]}t(v) \leq \max_{v' \in [0,1]} \widehat{U}_{1-\alpha}(v')\bigg)  &\geq 1 - \alpha\\
\hP\bigg(\min_{v\in [0,1]}t(v) \geq \min_{v' \in [0,1]} \widehat{L}_{1-\alpha}(v')\bigg)  &\geq 1 - \alpha
\label{eq:confidence-interval-95-2}
\end{align}
\end{lemma}

\begin{proof}[Proof of Lemma~\ref{lemma:validity-confidence-set}]

Recall that $\widehat{U}_{1-\alpha}(v)$ is a point-wise $100(1-\alpha)$\% level confidence interval.
This implies that
\begin{equation*}
\hP(t(v) \leq \widehat{U}_{1-\alpha}(v)) = 1 -  \alpha
\end{equation*}
for any $v \in [0, 1]$.
Now, let $v^{*} = \argmax_{v} t(v)$.
Then,  we have that
\begin{align*}
1- \alpha
= \hP(t(v^{*}) \leq \widehat{U}_{1- \alpha}(v^{*}))
\leq \hP(t(v^{*}) \leq \max_{v'}\widehat{U}_{1- \alpha}(v'))
\end{align*}
which proves that $\hP(\max_{v}t(v) \leq \max_{v'}\widehat{U}_{1- \alpha}(v')) \geq 1 - \alpha$.

\end{proof}

\subsection{Proofs}

\begin{proof}[Proof of Proposition~\ref{proposition:identification}]

Let $U$ denote a random variable with mean $0$ and variance $1$
and denote its cumulative distribution function by $F_{U}$.
For $v \sim \mathcal{U}(0, 1)$, we have
\begin{align*}
q_{0}(v)
&\equiv F_{Y^{*}_{00}} \circ F^{-1}_{Y^{*}_{01}}(v)\\
&= F_{U}
\bigg(
\frac{\mu_{01} - \mu_{00}}{\sigma_{00}} + \frac{\sigma_{01}}{\sigma_{00}}F_{U}^{-1}(v)
\bigg)
\end{align*}
The equality in the above expression holds because $Y^{*}_{dt}$ follows the location-scale family, which implies
\begin{align*}
F_{Y^{*}_{dt}}(y^{*})
&= F_{U}\bigg(
\frac{y^{*} - \mu_{dt}}{\sigma_{dt}}
\bigg)\\
F^{-1}_{Y^{*}_{dt}}(v)&=
\mu_{dt} + \sigma_{dt}F^{-1}_{U}(v)
\end{align*}

By Assumption~\ref{assumption:quantile},
\begin{align*}
F^{-1}_{Y^{*}_{11}}(v)
&= F^{-1}_{Y^{*}_{10}}(q_{0}(v))\\
&= \mu_{10} + \sigma_{10} F^{-1}_{U}(q_{0}(v))\\
&= \mu_{10} + \sigma_{10}
\bigg(
\frac{\mu_{01} - \mu_{00}}{\sigma_{00}} + \frac{\sigma_{01}}{\sigma_{00}}F_{U}^{-1}(v)
\bigg)\\
&\equiv \mu_{11} + \sigma_{11}F^{-1}_{U}(v)
\end{align*}
where
\begin{align*}
\mu_{11} \equiv \mu_{10} + \frac{\mu_{01} - \mu_{00}}{\sigma_{00}/ \sigma_{10}}
\quad
\text{and}
\quad
\sigma_{11} \equiv \frac{\sigma_{10}\sigma_{01}}{\sigma_{00}}.
\end{align*}

Combined with the fact that $U$ is a continuous and parametric distribution,
we recovers the distribution of $Y^{*}_{11}$.

\end{proof}

\begin{proof}[Proof of Proposition~\ref{proposition:validity-test}]
Case 1: $t \geq  \epsilon$.
In this case, the test makes a ``mistake'' because the upper bound does not cover $\epsilon$.
Thus, we can focus on an event $\{\max_{v}\widehat{U}_{1-\alpha}(v) < \epsilon\}$.
Since
$\hP(\max_{v}\widehat{U}_{1-\alpha}(v) < \epsilon)
= \hP(\widehat{U}_{1-\alpha}(v) < \epsilon, \forall v)$,
we can bound $\hP(\max_{v}\widehat{U}_{1-\alpha}(v) < \epsilon)$ as
\begin{equation}
\hP(\max_{v}\widehat{U}_{1-\alpha}(v) < \epsilon)
\leq
\min_{v}\hP(\widehat{U}_{1-\alpha}(v) < \epsilon)
\end{equation}
Now, consider a particular value of $v$.
Then, we have that
\begin{align*}
\hP(\widehat{U}_{1-\alpha}(v) < \epsilon)
&\leq
\hP(\widehat{U}_{1-\alpha}(v) < t)\\
&=
1 - \hP(t \leq \widehat{U}_{1-\alpha}(v))\\
& \leq
1 - (1 - \alpha) = \alpha\quad(n\to \infty)
\end{align*}
where the last inequality uses the fact that asymptotically $\widehat{U}_{1-\alpha}$ is a $(1-\alpha)$ level confidence interval (Lemma~\ref{lemma:validity-confidence-set}).

Case 2: $t \leq  - \epsilon$.
In this case, we focus on the other event $\{\min_{v}\widehat{L}_{1-\alpha}(v) \geq - \epsilon \}$.
Since $\hP(\inf_{v}\widehat{L}_{1-\alpha}(v) \geq - \epsilon) \leq \min_{v}
\hP(\widehat{L}_{1-\alpha}(v) \geq - \epsilon)$, we have that
\begin{align*}
\hP(\widehat{L}_{1-\alpha}(v) \geq - \epsilon)
&\leq
\hP(\widehat{L}_{1-\alpha}(v) \geq t)\\
&=
1 - \hP(\widehat{L}_{1-\alpha}(v) \leq t)\\
&\leq
1 - (1- \alpha) = \alpha\quad(n\to \infty)
\end{align*}
where the last inequality is a direct application of Lemma~\ref{lemma:validity-confidence-set}.
\end{proof}

\section{Extensions}\label{appendix:extensions}
\setcounter{figure}{0}

\subsection{Other estimands}\label{sec:extensions}
This section provides an extension of the proposed methodology.
Following the recent developments in the literature on causal inference with ordinal outcome,
where more interpretable estimands have been proposed,
I show how to apply the proposed method to those new estimands.

In the above section, we have focused on particular causal estimands, $\zeta_{j}$ and $\Delta_{j}$, which is a difference in probabilities defined for a  specific reference category $j$.
One issue of this quantity $\Delta_{j}$  is that depending on the choice of  reference category $j$,
the sign of the estimate might flip.
This means that interpretation becomes tricky because it is completely possible to observe $\Delta_{j} > 0$ and $\Delta_{j'} < 0$ for $j \neq j'$  with the same data.
From this, we cannot even conclude that the treatment had ``positive'' effect or not.

To circumvent this problem associated with $\Delta_{j}$,
recent papers turn to different kinds of estimands for ordinal outcome
\citep[e.g.,][]{volfovsky2015causal, chiba2017sharp,lu2018treatment, lu2018partial}.
For example, \cite{lu2018treatment} considers the following estimand:
\begin{equation}
\eta = \hP(Y_{i1}(1) \geq Y_{i1}(0) \mid D_{i} = 1)
\end{equation}
This is a proportion units of who benefit from (or at least not harmed by) the treatment.
In our example, $\eta$ captures the proportion of treated respondents
who change their opinion toward gun control (regardless of their baseline attitudes)
after experiencing the mass shooting in their neighborhood.
This quantity is easy to interpret because it does not depend on the baseline attitude
and smaller value of $\eta$ indicates that there are few  respondents who change their opinion towards gun controls.
However, since $\eta$ depends on the joint distribution of potential outcomes, $(Y_{i1}(0), Y_{i1}(1))$, it cannot be point identified.
\cite{lu2018treatment} provides a closed form bound for this estimand using only the marginal distribution of $Y_{i1}(1)$ and $Y_{i1}(0)$.

The benefit of the proposed method over the dichotomizing-the-outcome approach is that
it identifies the entire distribution of $Y_{i1}(0) \mid D_{i} = 1$,
whereas the information about the entire distribution is lost when we use the coarsened outcome.
This implies that we can estimate the bound based on $\widehat{\theta}_{11}$ given in Equation~\eqref{eq:second-stage-theta}.
Following the result of \cite{lu2018treatment}, we can estimate the bound $[\widehat{\underline{\eta}}, \widehat{\overline{\eta}}]$ as
\begin{align*}
\widehat{\underline{\eta}}
= \max_{0\leq j \leq J-1}
\Big\{
[\Phi(\kappa_{j+1}\mid \widehat{\theta}_{11}) - \Phi(\kappa_{j}\mid \widehat{\theta}_{11})] + \widehat{\Delta}_{j}
\Big\}\quad
\text{and}\quad
\widehat{\overline{\eta}} = 1 +
\min_{0 \leq j \leq J-1} \widehat{\Delta}_{j}
\end{align*}
where $\widehat{\Delta}_{j}$ is the estimate of the distributional effect,
and $\widehat{\Delta}_{0} = 0$ by construction.
We can see from the formula that the upper bound is informative
as long as there is at least one reference category $j$ such that $\widehat{\Delta}_{j} < 0$ for $j = 1, \ldots, J-1$.
Otherwise, $\Delta_{0} = 0$ will be the minimum and thus we get  the non-informative upper bound $\widehat{\overline{\eta}} = 1$.


\subsection{Time-varying covariates}
Researchers might want to incorporate time-varying covariates
into the analysis to further gain efficiency.
I discuss that the parametric specification of the proposed method
allows the use of such covariates for analysis.
Although sometimes appealing, I emphasize that
parametric specification introduces additional assumptions for the analysis.

Let $\mathbf{X}_{it} \in \R^{p}$ denote a $p$
dimensional vector of time varying covariates.
We can model the mean and the variance of the latent utilities as
\begin{equation*}
\mu_{it} = \*Z^{\top}_{it}\bm{\gamma}_{0}
\quad \text{and}\quad
\sigma_{it} = \exp(\*Z^{\top}_{it}\bm{\gamma}_{1})
\end{equation*}
where $\mathbf{Z}_{it} = (1, D_{i}, t, D_{i}\cdot t, \mathbf{X}^{\top}_{it})^{\top}$.
Then, we can express the observed choice probability as
\begin{equation*}
\hP(Y_{it} = j\mid \*Z_{it}) =
  \Phi(\kappa_{j+1} \mid \mu_{it}, \sigma_{it}) -
  \Phi(\kappa_{j} \mid \mu_{it}, \sigma_{it}).
\end{equation*}


We estimate parameters $\bm{\gamma} = (\bm{\gamma}^{\top}_{0}, \bm{\gamma}^{\top}_{1})^{\top}$ by the maximum likelihood.
\begin{equation*}
\widehat{\bm{\gamma}}  = \argmax_{\bm{\gamma}}\sum^{n}_{i=1}\sum^{1}_{t=0}\sum^{J-1}_{j=0}\mathbf{1}\{Y_{it} = j\}\log
\Big\{
\Phi(\kappa_{j+1}\mid \mathbf{Z}_{it}, \bm{\gamma}) - \Phi(\kappa_{j}\mid \mathbf{Z}_{it}, \bm{\gamma})
\Big\}.
\end{equation*}

Finally, the quantities of interest is estimated
by taking the sample average of predicted probabilities.
\begin{equation*}
  \widehat{\Delta}_{j} =
  \frac{1}{n_{1}}\sum^{n}_{i=1}D_{i}\Big\{
    \hP(Y_{i1} \geq j \mid D_{i} = 1, \*X_{i1}, \widehat{\bm{\gamma}}) - \hP(Y_{i1} \geq j \mid D_{i} = 0, \*X_{i1}, \widehat{\bm{\gamma}})
  \Big\}
\end{equation*}
Note that the marginalization of covariates is with respect to the distribution for the treated,
because our estimand $\Delta_{j}$ is defined for the treated units.

\section{Dichotomizing the Outcome: An Example}
\label{appendix:binarize-outcome}
\setcounter{figure}{0}

Coarsening the ordinal outcome into a binary variable is a common practice often employed in applied works.
Although this procedure allows scholars to utilize the standard linear DID,
I will demonstrate in this section that this operation leads to an inconsistent result depending on how the new variable is created.

To see this, let's consider a simple example with three categories, $Y_{it} \in \{0, 1, 2\}$.
There are two possible ways to transform this variable into a binary outcome, $\widetilde{Y}_{it} = \mathbf{1}\{Y_{it} = 2\}$ or $\check{Y}_{it} = \mathbf{1}\{Y_{it} \geq 1\}$.
Under this setup, we require two separate parallel trends assumptions for identification,
\begin{align*}
&\text{PT1}\colon \E[\widetilde{Y}_{i1}(0) - \widetilde{Y}_{i0}(0) \mid D_{i} = 1] = \E[\widetilde{Y}_{i1}(0) - \widetilde{Y}_{i0}(0) \mid D_{i} = 0]\\
&\text{PT2}\colon \E[\check{Y}_{i1}(0) - \check{Y}_{i0}(0) \mid D_{i} = 1] = \E[\check{Y}_{i1}(0) - \check{Y}_{i0}(0) \mid D_{i} = 0]
\end{align*}

PT1 identifies $\Delta_{2} = \Pr(Y_{i1}(1)  = 2 \mid D_{i} = 1) - \Pr(Y_{i1}(0)  = 2 \mid D_{i} = 1)$
and PT2 identifies $\Delta_{1} = \Pr(Y_{i1}(1)  \geq 1 \mid D_{i} = 1) - \Pr(Y_{i1}(0)  \geq 1 \mid D_{i} = 1)$.
Also let $\pi^{(t)}_{j|d} = \Pr(Y_{it}(0) = j \mid D_{i} = d)$ be the conditional probability for the potential outcome under the control.

Now consider the following data generating process which specifies the marginal distributions
for the potential outcome:
\begin{align*}
\left(\pi^{(0)}_{j=0|d = 1}, \pi^{(0)}_{j=1|d= 1}, \pi^{(0)}_{j=2|d=1}\right) &= (0.3, 0.5, 0.2)\\
\left(\pi^{(1)}_{j=0|d = 1}, \pi^{(1)}_{j=1|d= 1}, \pi^{(1)}_{j=2|d=1}\right) &= (0.2, 0.5, 0.3)\\
\left(\pi^{(0)}_{j=0|d = 0}, \pi^{(0)}_{j=1|d= 0}, \pi^{(0)}_{j=2|d=0}\right) &= (0.2, 0.5, 0.3)\\
\left(\pi^{(1)}_{j=0|d = 0}, \pi^{(1)}_{j=1|d= 0}, \pi^{(1)}_{j=2|d=0}\right) &= (0.2, 0.4, 0.4)
\end{align*}

Under this DGP, PT1 holds since
\begin{align*}
&\E[\widetilde{Y}_{i1}(0) - \widetilde{Y}_{i0}(0) \mid D_{i} = 1] - \E[\widetilde{Y}_{i1}(0) - \widetilde{Y}_{i0}(0) \mid D_{i} = 0]\\
&= [\pi^{(1)}_{2|1} - \pi^{(0)}_{2|1}] - [\pi^{(1)}_{2|0} - \pi^{(0)}_{2|0}]\\
&= 0.1 - 0.1 = 0.
\end{align*}
However, PT2 does not hold because
\begin{align*}
&\E[\check{Y}_{i1}(0) - \check{Y}_{i0}(0) \mid D_{i} = 1] - \E[\check{Y}_{i1}(0) - \check{Y}_{i0}(0) \mid D_{i} = 0]\\
&= \Big\{[\pi^{(1)}_{2|1} + \pi^{(1)}_{1|1}] - [\pi^{(0)}_{2|1} + \pi^{(0)}_{1|1}] \Big\} -
\Big\{[\pi^{(1)}_{2|0} + \pi^{(1)}_{1|0}] - [\pi^{(0)}_{2|0} + \pi^{(0)}_{1|0}]\Big\}\\
&= \{(0.3 + 0.5) - (0.2 + 0.5)\} - \{(0.4 + 0.4) - (0.3 + 0.5)\} = 0.1.
\end{align*}

Thus, this example demonstrate that
with the same data, $\Delta_{2}$ can be consistently estimated with $\widetilde{Y}_{it}$
but $\Delta_{1}$ cannot be estimated without bias,
even though we have the same data generating process behind the two transformations.
Obviously, it is also trivial to construct an example where PT2 holds but PT1 does not.

\section{Additional Empirical Results}
\label{appendix:subset}
\setcounter{figure}{0}

\subsection{Different coding of partisanship}

Given the partisan nature of the gun control policies,
it is important to understand if the effect of mass shootings could differ by respondents' party identification.
The coding of partisanship, however, slightly different across studies.
In the main text, I relied on the 3-point scale party identification question (\texttt{pid3}),
which asks respondents the following question,
\begin{quote}
\texttt{Generally speaking, do you think of yourself as a ...?}

\texttt{
(1) Democrat;
(2) Republican;
(3) Independent;
}

\texttt{
(4) Other;
(5) Not sure;
(8) Skipped.
}
\end{quote}
In the analysis, I exclude respondents who do not select option \texttt{(1)}, \texttt{(2)} or \texttt{(3)}.

\begin{figure}[htb]
  \centerline{\includegraphics[scale=0.8]{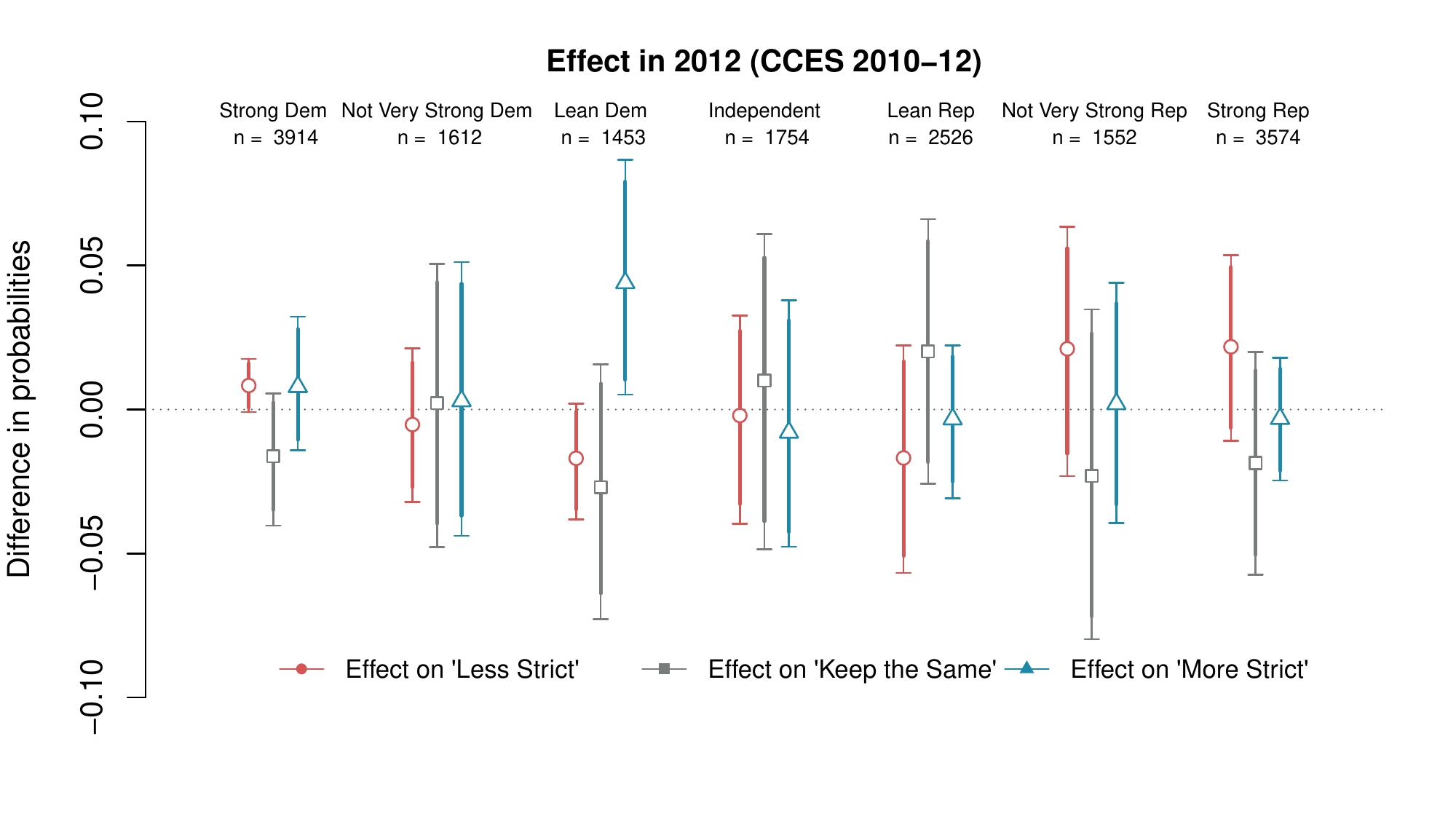}}
  \caption{Estimated effect based on a 7-point scale partisanship (\texttt{pid7}).}
  \label{fig:twowaves-full-pid7}
\end{figure}

In the survey, respondents are also asked to place themselves
on a granular scale of partisanship (\texttt{pid7}):
\begin{quote}
\texttt{
  Would you call yourself a strong Democrat or a not very strong Democrat? Would you call yourself a strong Republican or a not very strong Republican? Do you think of yourself as
  closer to the Democratic or the Republican Party?
}

\texttt{
(1) Strong Democrat;
(2) Not very strong Democrat;
(3) Lean Democrat;
}

\texttt{
(4) Independent;
(5) Lean Republican;
(6) Not very strong Republican;
}

\texttt{
(7) Strong Republican;
(8) Not sure;
(98) Skipped.
}
\end{quote}
Figure~\ref{fig:twowaves-full-pid7} shows the result of the analyses that use \texttt{pid7} to construct partisan subgroups.
We can see that the effect is concentrated among ``Lean Democrats.''

On the other hand, \cite{barney2019reexamining} constructs the partisanship variable based on \texttt{pid7} but collapses it into a 3-point scale, ``Democrat'', ``Independent'' and ``Republican''.
The major difference from the self-reported \texttt{pid3}
is that ``leaners'' are classified as partisans (i.e., not independents).
\begin{figure}[htb]
  \centerline{\includegraphics[scale=0.8]{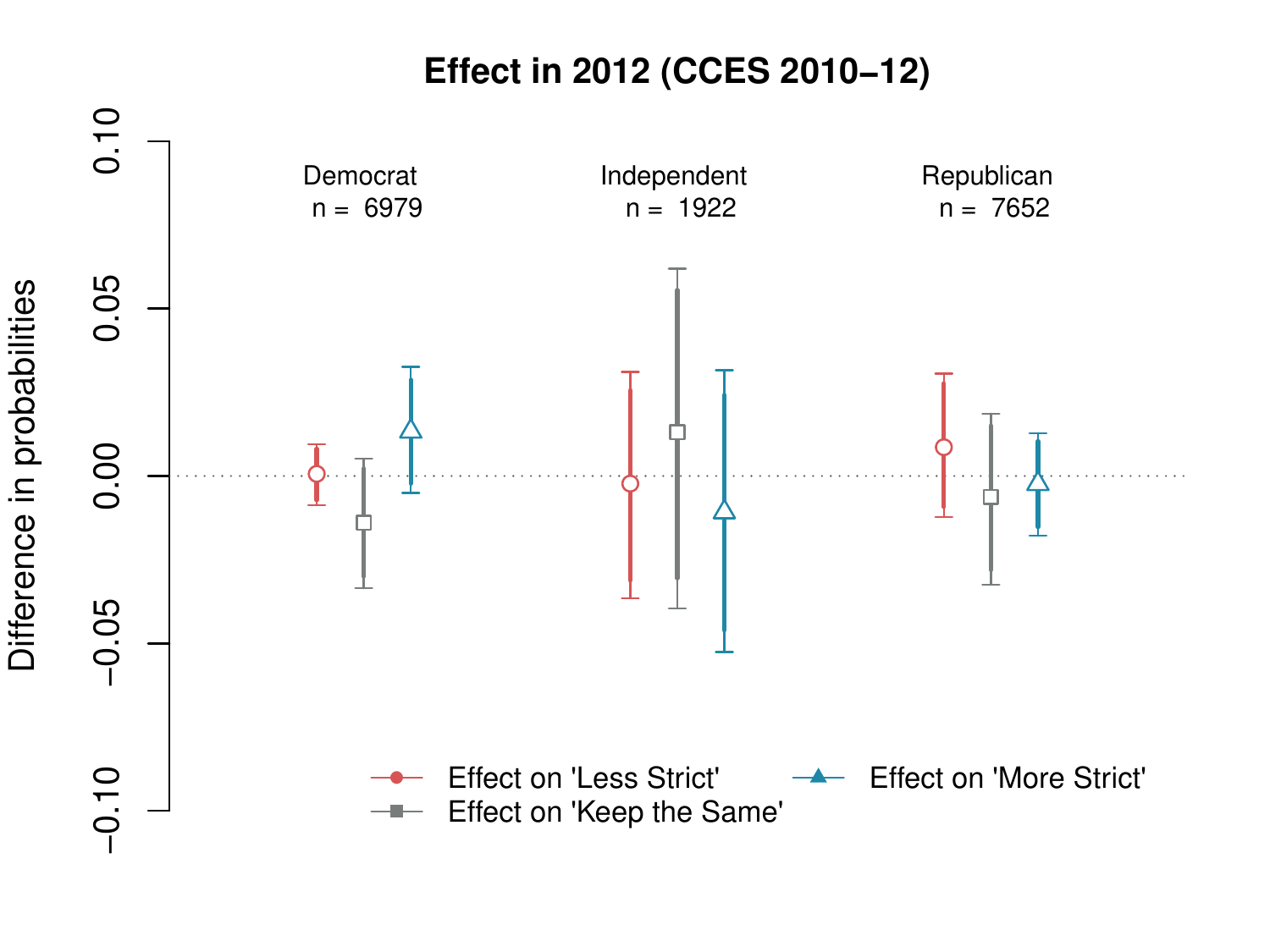}}
  \centerline{\includegraphics[scale=0.8]{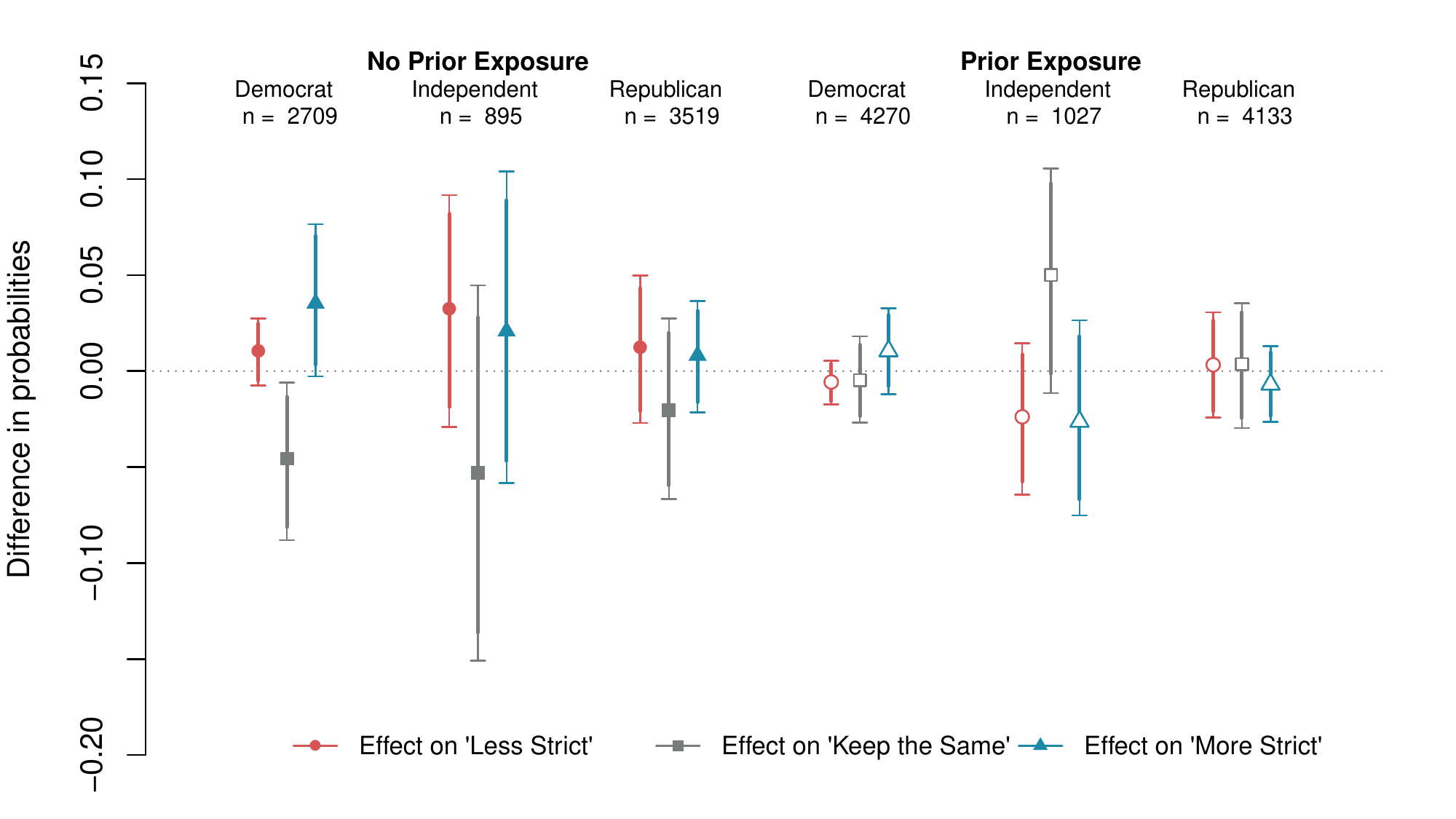}}
  \caption{Estimated effects based on partisanship based on a coding used in \cite{barney2019reexamining}.}
  \label{fig:twowaves-full-party2010}
\end{figure}
Figure~\ref{fig:twowaves-full-party2010} reports the estimate
based on the partisan coding of \cite{barney2019reexamining}.

%
%
%

\subsection{Different estimands}

Figure~\ref{fig:other-estimands} shows estimated bounds on $\tau = \Pr(Y_{i1}(1) \geq Y_{i1}(0) \mid D_{i} = 1)$ (gray lines) and $\eta = \Pr(Y_{i1}(1) > Y_{i1}(0) \mid D_{i} = 1)$ (red lines).
The explicit formula of the bound for each estimand is given in Section~\ref{sec:extensions}.
\begin{figure}[ht]
    \centerline{\includegraphics[scale=0.8]{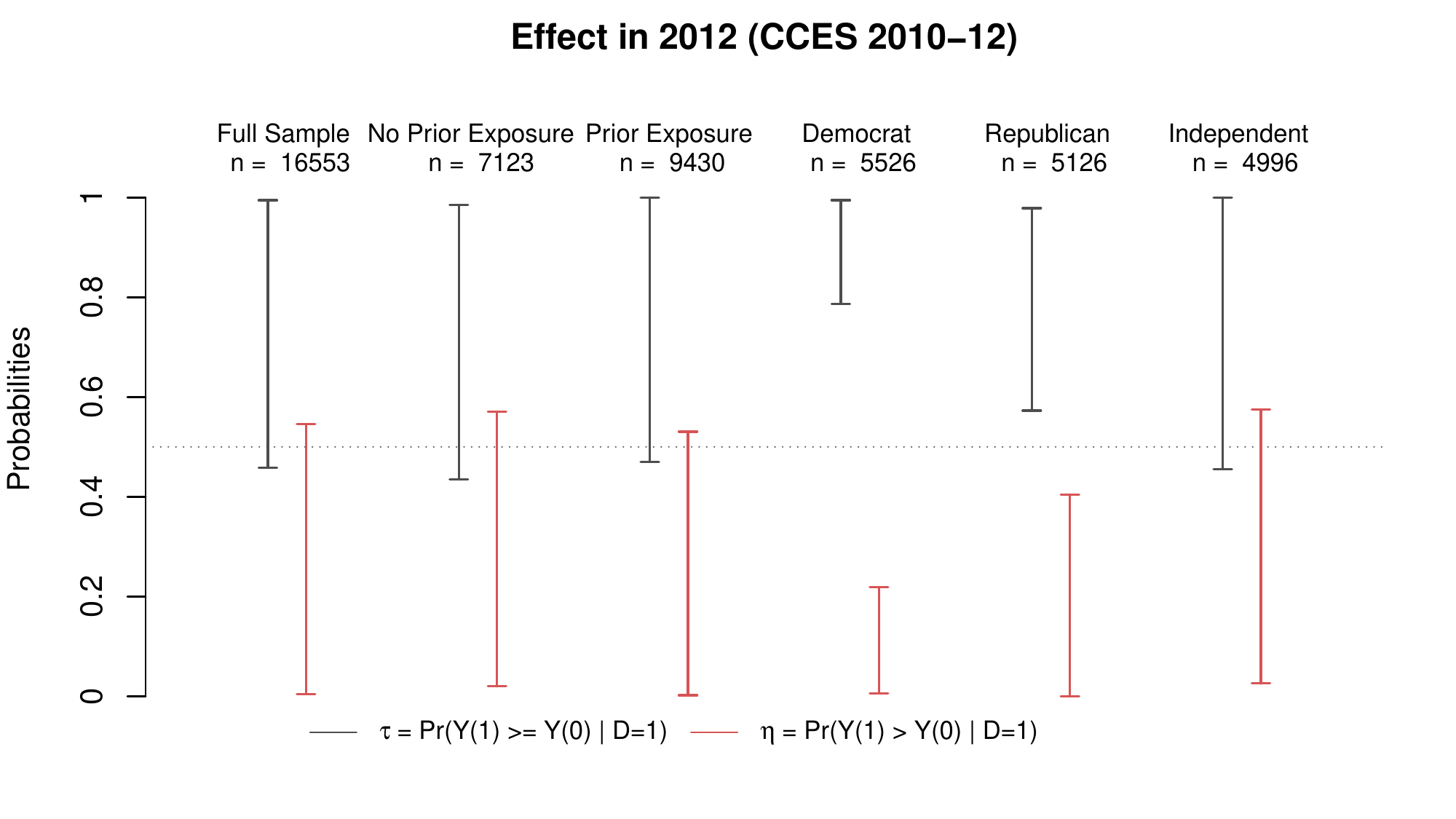}}
  \caption{Estimated bound for $\tau$ (gray lines) and $\eta$ (red lines).}
  \label{fig:other-estimands}
\end{figure}
I find that bounds for the two estimands diverge suggesting there are many observations that have $Y_{i1}(1) = Y_{i1}(0)$ in the population.
For example, among Democrats, the bound for $\tau$ is between around 0.8 and 1.0
which might suggest that proportion of Democrats who supports gun control when treated is extremely high.
However, the bound for $\eta$ is between around 0.2 and 0 which might suggest that proportion of Democrats who
have \textit{strictly} prefer a more strict gun control is very small.
This is a problem discussed in \cite{lu2018treatment} where $\tau$ and $\eta$ cannot be informative
when there are many units who does not change attitudes by the treatment.
Therefore, it appears that we need to turn to other estimators that avoid this problem \citep[e.g.,][]{chiba2017sharp,lu2018partial}.

\subsection{Different treatment cutoff}

Although \cite{newman2019mass} define the exposure by the 100-mile cutoff, \cite{barney2019reexamining} consider different threshold to assess the robustness of the results.
Following their analysis,  I consider an alternative threshold
of 25 miles.
Figure~\ref{fig:twowave-main-25mi} shows the result.
Note that the ``prior exposure'' is also defined by the 25-mile cutoff.
\begin{figure}[ht]
  \centerline{
    \includegraphics[scale=0.8]{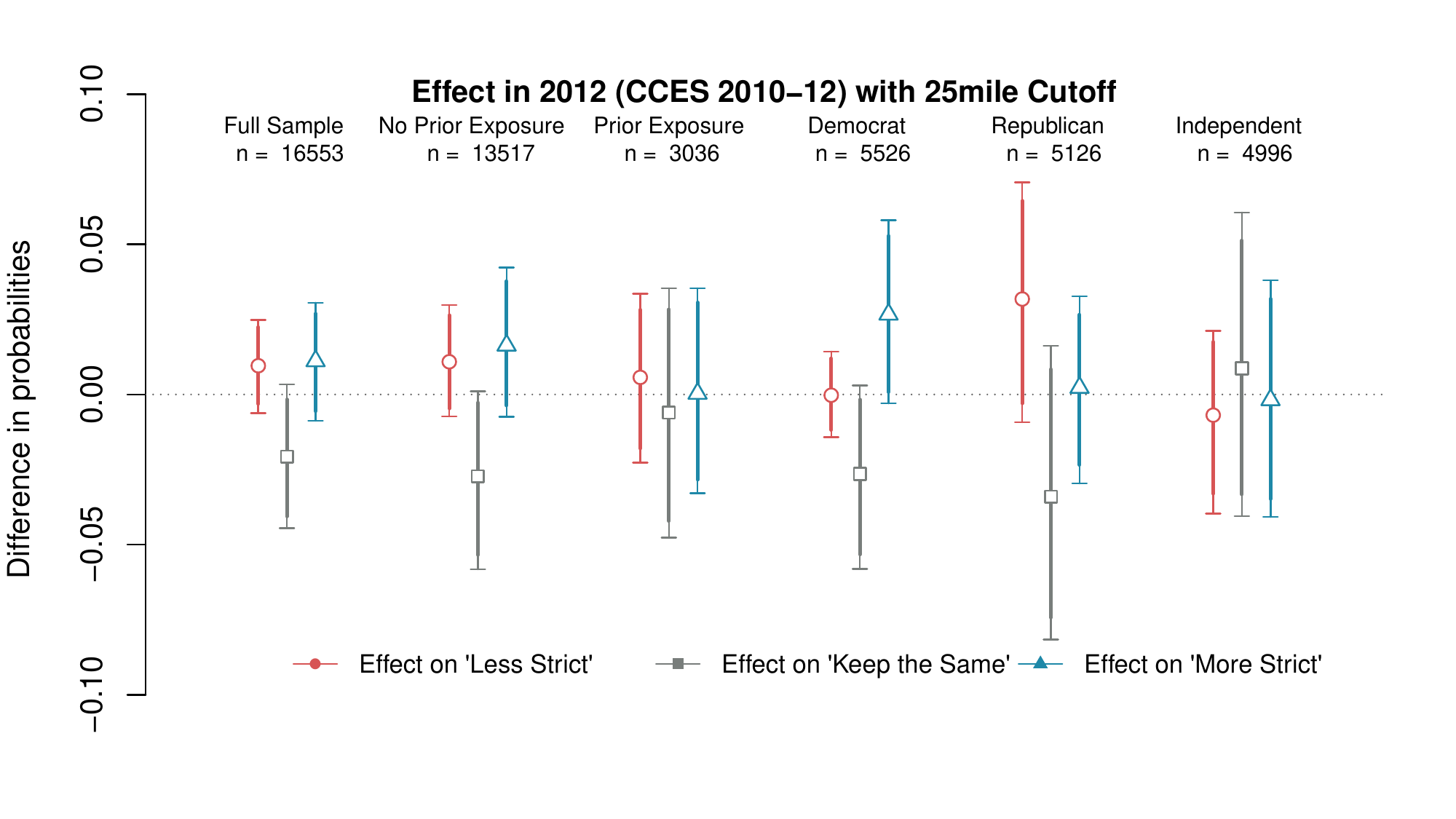}
  }
  \centerline{
    \includegraphics[scale=0.8]{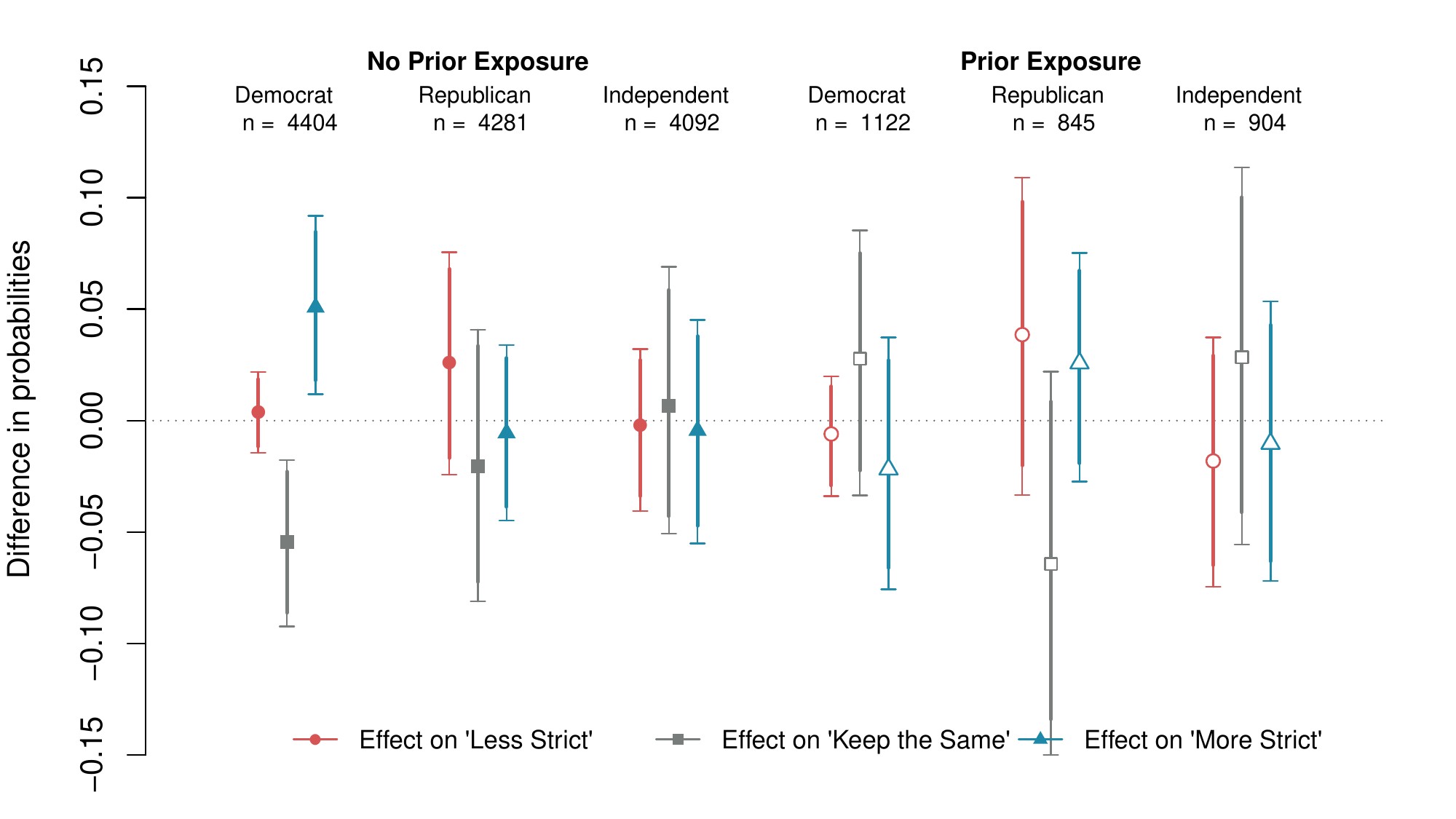}
  }
  \caption{Estimated causal effect with \textbf{25-mile} cutoff as the threshold for the exposure. Circles are the estimate of $\zeta_{0}$, square are the estimate of $\zeta_{1}$ and triangles are the estimate of $\zeta_{2}$. Thin (thick) lines indicate 90\% (95\%) confidence intervals.}
  \label{fig:twowave-main-25mi}
\end{figure}

%

\subsection{Two year subset of three-wave panel}

In this section, I present two additional results based on the two-wave panel from CCES.
Figure~\ref{fig:twowaves-subset} shows estimated effects based on the sub-group analysis taking interactions between the prior exposure variable and partisan identification.
In the figure, circles (triangles) show estimates of $\Delta_{2}$ ($\Delta_{1}$) and thin (thick) line indicate 90\% (95\%) confidence intervals.
Confidence intervals are computed based on block bootstraps (blocked at the zip code level).
\begin{figure}[h]
  \centerline{\includegraphics[scale=0.8]{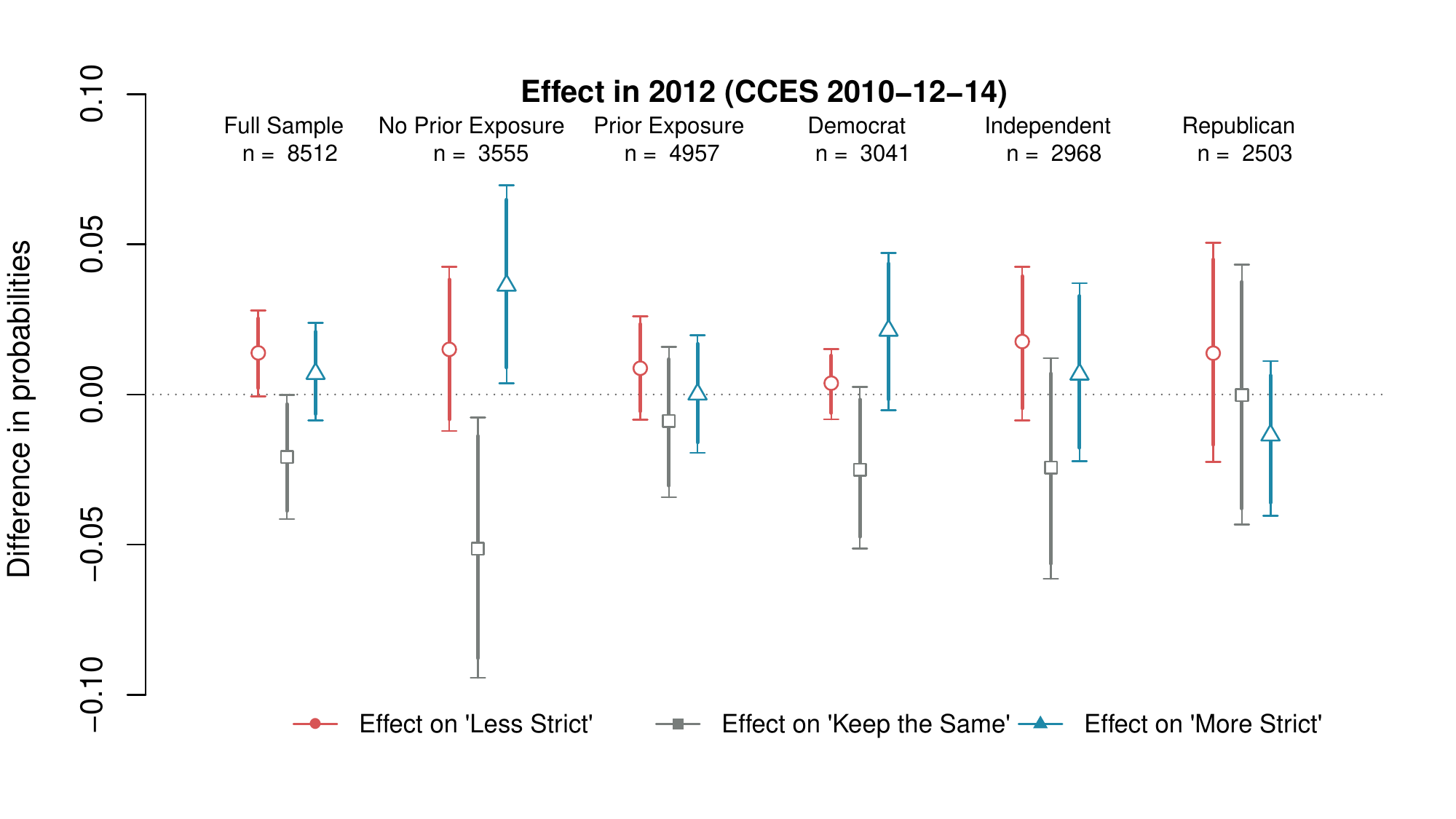}}
  \centerline{\includegraphics[scale=0.8]{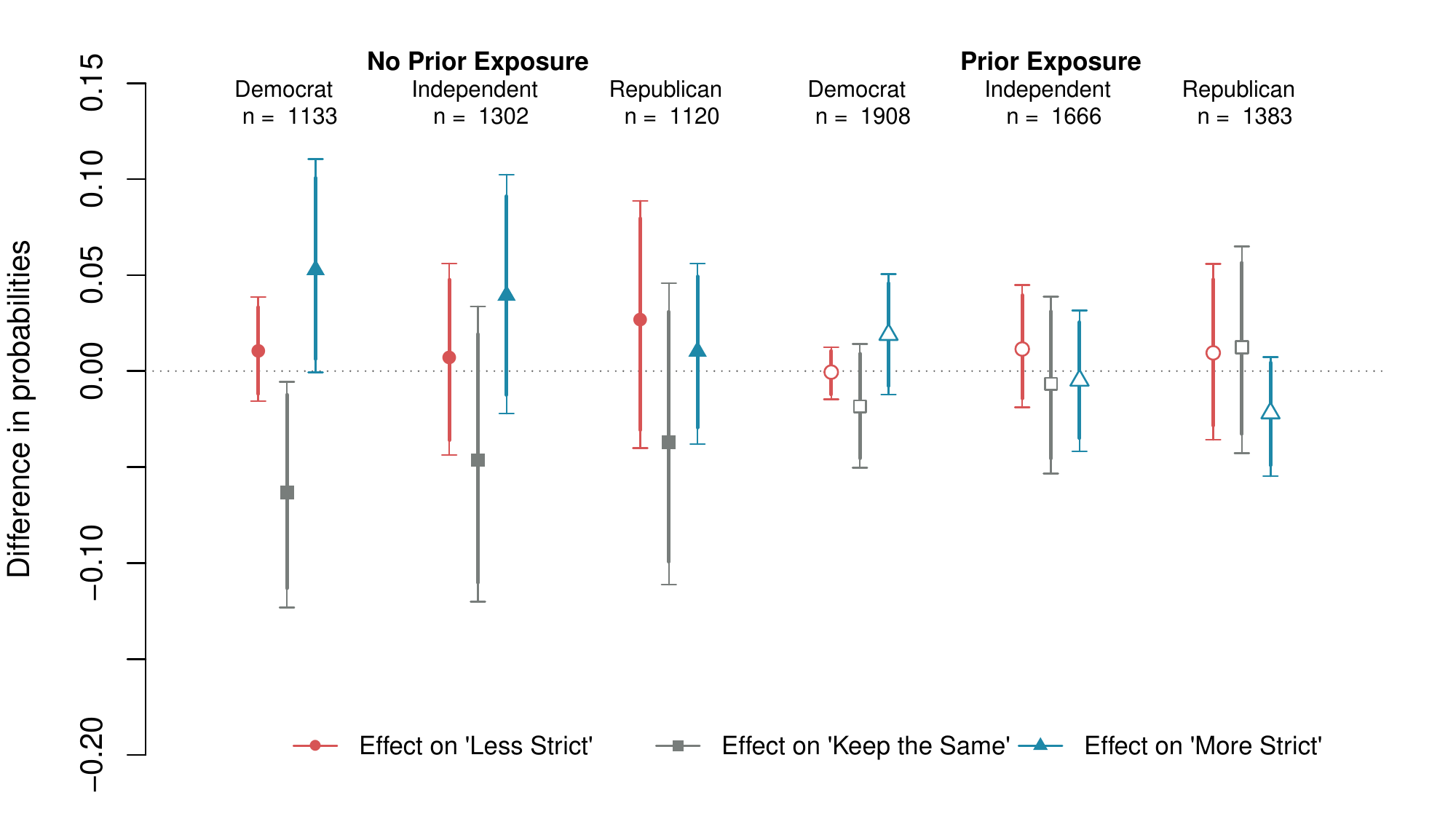}}
  \caption{Estimated effects using the two-year subset of the three-year panel. Circles are the estimate of $\zeta_{0}$, square are the estimate of $\zeta_{1}$ and triangles are the estimate of $\zeta_{2}$. Thin (thick) lines indicate 90\% (95\%) confidence intervals.}
  \label{fig:twowaves-subset}
\end{figure}
I find that among no-prior-exposure group, the effect is concentrated among Democrats ($\Delta_{2}$ for Democrats is estimated positive and statistically different from zero at the 10\% level, while $\Delta_{1}$ is not statistically significant).
On the other hand, effects for Independents and Republicans are both indistinguishable from zero at the 10\% level (neither $\Delta_{2}$ nor $\Delta_{1}$).
I also find that effects are almost zero in the prior-exposure groups regardless of partisanship.

\section{Details of the Application}
\label{sec:detail-application}
\setcounter{figure}{0}

\paragraph{A list of method used in the original studies}

Table~\ref{tab:summary-method} summarizes the methods used in the original papers.

\begin{table}[ht]
	\centering
		\caption{Methodologies used in the original studies.
    Abbreviation: \cite{newman2019mass} as NH19, \cite{barney2019reexamining} as BS19
    and  \cite{hartman2019accounting} as HN19.}
		\label{tab:summary-method}
	\begin{tabular}{cccc}
		\toprule
						& NH19  & BS19 &  HN19\\
		\midrule
		\texttt{ordered logit (RE)} & \checkmark (with Lag DV) & \checkmark & \checkmark \\
		\texttt{ordered logit (FE)} & & & \checkmark \\
		\texttt{linear two-way FE}         & & \checkmark & \\
		\bottomrule
	\end{tabular}
\end{table}

\paragraph{Coding of mass shootings}
\cite{newman2019mass} uses the following criteria to determine if an incident constitutes a mass public shooting: ``(1) firearms as the primary weapon used, (2) attacks on non-family members of the general public and (3) attacks in which at least three or more individuals were injured or killed.'' \citep[][p.8]{newman2019mass}. See the original studies for the detail of why these criteria are selected.
Note that the definition of the ``treatment'' is slightly different between \cite{newman2019mass} and \cite{barney2019reexamining}. I follow the definition used by \cite{barney2019reexamining}; Please see \cite{barney2019reexamining} for the discussion on this point.

\paragraph{Survey outcome}
The ordering of the response categories is not exactly the same as the original question in CCES 2010--2012 panel. Originally in the survey, the choices are given as (1) More Strict; (2) Less Strict; (3) Kept As They Are (please see \texttt{CC10\_320} and \texttt{CC12\_320} in ``Guide to the 2010-12 CCES Panel Study'' available at \url{https://doi.org/10.7910/DVN/24416/79YKV2}).
In the main text, I follow the coding of \cite{newman2019mass} and \cite{barney2019reexamining} and treat ``Kept As They Are'' as the middle category.

Figure~\ref{fig:outcome-distribution} shows the distribution of the outcome in 2010 (top) and 2012 (bottom) where the blue bars correspond to the treatment group
and the gray bars correspond to the control group.
\begin{figure}[h]
  \centerline{\includegraphics[scale=0.8]{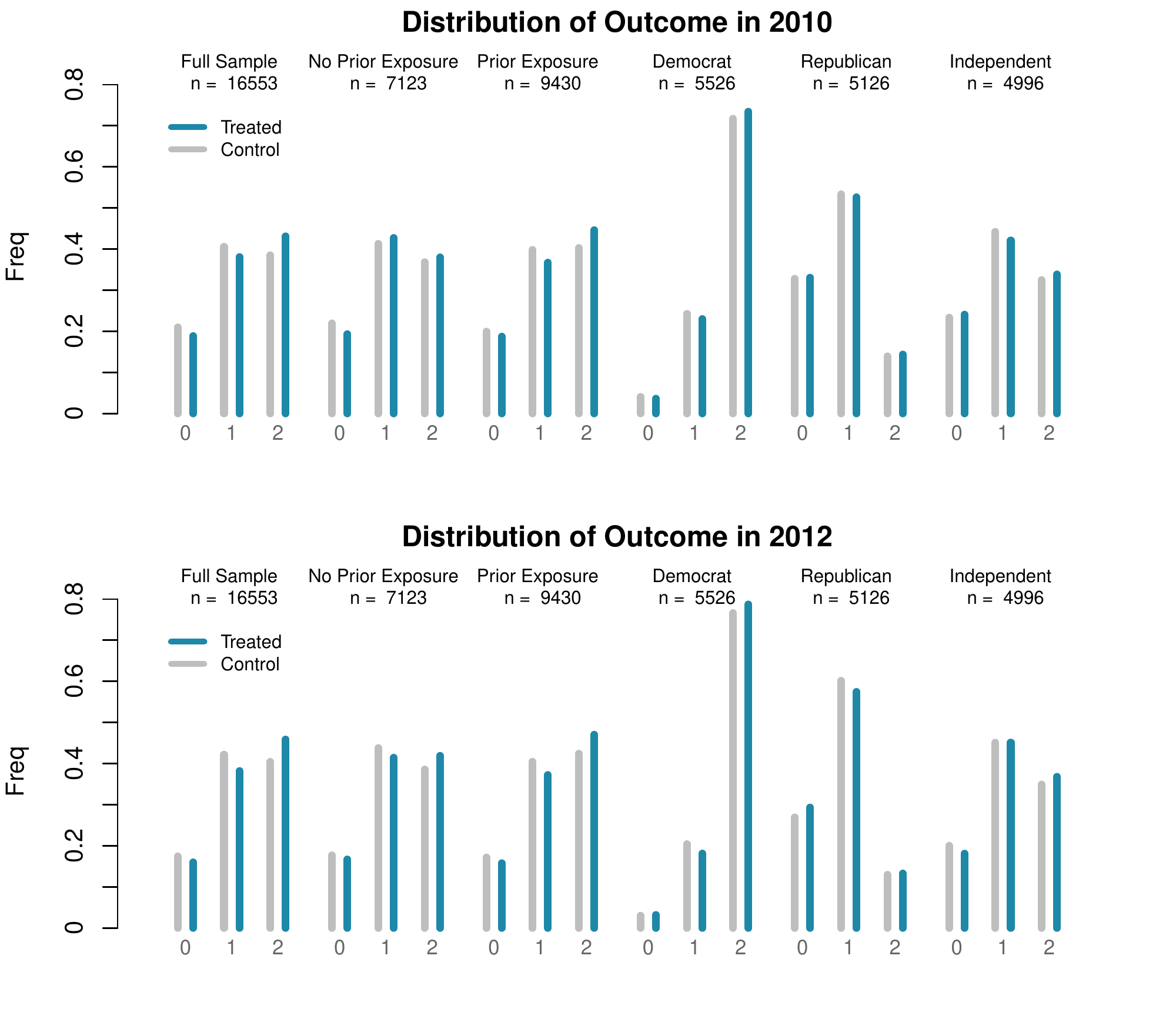}}
  \caption{Distribution of outcomes: \texttt{(0): less-strict}, \texttt{(1): kept-as-they-are} and \texttt{(2): more-strict}. The top panel shows the distribution of 2010 and the bottom panel is for 2012. Bars in blue (gray) shows distributions for the treated (control) group.}
  \label{fig:outcome-distribution}
\end{figure}

\section{Simulation Studies}\label{sec:simulation}
\setcounter{figure}{0}


In this section, I present two Monte Carlo studies to investigate finite sample performances of the proposed method.
The first simulation assesses performance of the proposed estimator for the causal effect
where  I compare the proposed estimator against the standard difference-in-differences with dichotomized outcome and the ordered probit regression.
The result shows that the proposed estimator is unbiased to the causal effects and the confidence interval has nominal coverage,
while the other two method are biased and thus the confidence intervals fail to maintain the coverage.
The second simulation studies the finite sample performance of the proposed procedure for the diagnostic in Section~\ref{sec:diagnostics}.
I demonstrate that the type I error is controlled under the range of equivalence threshold  that is compatible with the null
and that the power converges to one when the equivalence threshold is chosen reasonably.

\subsection{Estimating causal effects}

In this first simulation study, I investigate the finite sample performance of the proposed estimator
under the correct model specification.
The potential outcome is generated by first drawing the latent utilities from the normal distribution.
For the potential outcome under the control, the following set of parameters are used to generate the data:
$\theta_{00} = (-0.5, 1.5)^{\top}$, $\theta_{01} = (1, 1)^{\top}$ and $\theta_{10} = (-1.5, 2)^{\top}$.
The parameters for the counterfactual outcome $\theta_{11}$ is set according to the identification formula in Proposition~\ref{proposition:identification}.
The parameters for generating the potential outcome under the treatment, $Y_{i1}(1)$, are set as $\mu = 1.5$ and $\sigma = 1.5$.

After generating the latent utilizes,
they are transformed into categorical outcome
with $J$ categories
based on the set of cutoffs.
In this simulation, I consider $J \in \{3, 5, 7\}$
and also I vary the number of units $n \in \{1000, 2500, 5000\}$.

Since there are $J - 1$ possible treatment effects to consider, that is, $\{\Delta_{j}\}^{J-1}_{j=1}$,
estimators are evaluated on averaging the loss over $J - 1$ treatment effect estimates.
Specifically, I consider the following metrics:
\begin{align*}
\overline{\mathsf{Abs.Bias}}
&= \frac{1}{(J-1)}\sum^{J-1}_{j=1}\bigg|\frac{1}{S}\sum^{S}_{s=1}(\widehat{\Delta}^{(s)}_{j}  - \Delta_{j})\bigg|\\
\overline{\mathsf{RMSE}}
&= \frac{1}{(J-1)}\sum^{J-1}_{j=1}\bigg\{\frac{1}{S}\sum^{S}_{s=1}(\widehat{\Delta}^{(s)}_{j}  - \Delta_{j})^{2}\bigg\}^{1/2}\\
\overline{\mathsf{Coverage}}
&= \frac{1}{(J-1)S} \sum^{J-1}_{j=1}\sum^{S}_{s=1}\mathbf{1}\Big\{\Delta_{j} \in \widehat{C}^{(s)}_{j,1-\alpha/2}\Big\}
\end{align*}
where $\widehat{\Delta}^{(s)}_{j}$ is the estimate of $\Delta_{j}$ under $s$th Monte Carlo iteration
and $\widehat{C}^{(s)}_{j, 1 - \alpha/2}$ is the $100\times (1 - \alpha/2)$\% confidence interval for $\Delta_{j}$.

\begin{figure}[h]
  \centerline{\includegraphics[scale=1.2]{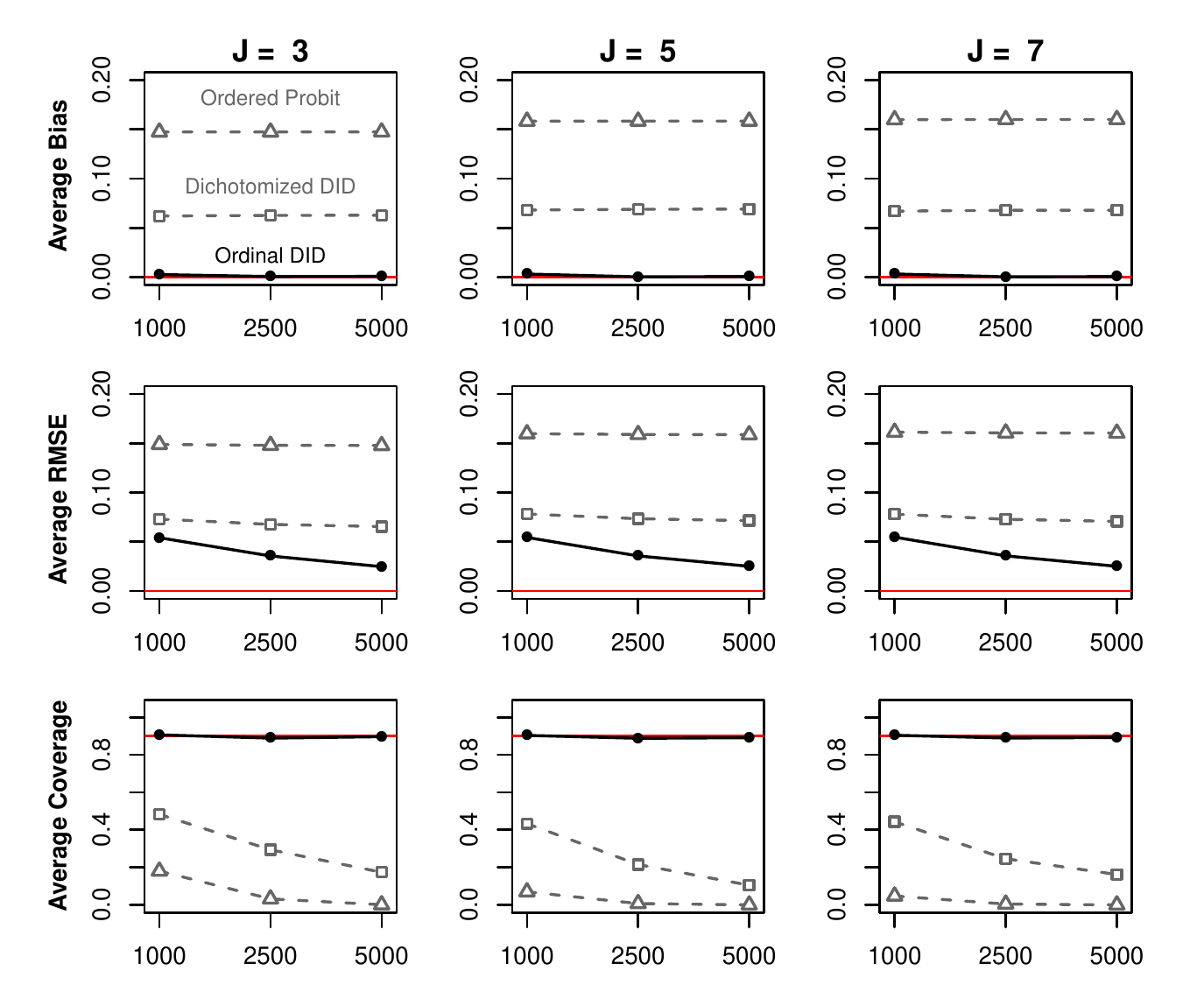}}
  \caption{Simulation Results. Top row: Absolute bias ($\overline{\textsf{Abs. Bias}}$). Middle: RMSE ($\overline{\textsf{RMSE}}$). Bottom: Coverage based on the 90\% confidence interval ($\overline{\textsf{Coverage}}$).
          As expected from the general result of Maximum Likelihood, the estimate is unbiased and the confidence interval maintains nominal coverage under the correct specification.}
  \label{fig:sim-point-estimate}
\end{figure}

Figure~\ref{fig:sim-point-estimate} shows the result.
Left panel shows the absolute bias of the estimate.
We see that the bias is larger when the sample is relatively small for $J = 5$ as the number of observations
in each category tend to be smaller.
However, in general, estimates are unbiased.
Middle panel shows the RMSE.
It shows that RMSE decreases as the sample size increases and the variance is smaller when the number of categories are smaller.
Finally, the right panel shows the coverage of 90\% confidence intervals.
We can see that for both cases, confidence intervals have nominal coverage regardless of sample size.

\subsection{Testing procedure}
In this section, I investigate a finite sample performance of the proposed testing procedure.
Specifically, I conduct a Monte Carlo simulation with a scenario that Assumption~\ref{assumption:quantile}
is violated in the pre-treatment periods.
Outcomes are generated first by simulating the latent utilities.
The latent utilities are simulated according to the normal distribution with mean $\mu_{dt}$
and variance $\sigma^{2}_{dt}$,
\begin{equation}
Y^{*}_{dt} \sim \mathcal{N}(\mu_{dt}, \sigma^{2}_{dt})
\end{equation}
I set $\theta_{00} = (-0.5, 1.5)^{\top}$, $\theta_{01} = (1, 1)^{\top}$ , $\theta_{10} = (-1.5, 2)^{\top}$  and $\theta_{11} = (1.5, 1.5)^{\top}$.
This parameter specification leads to the true maximum deviation $t_{\max} \approx 1.4$.
Clearly, this does not satisfy Assumption~\ref{assumption:quantile} which requires $t_{\max} = 0$.

After simulating $Y^{*}_{dt}$, categorial outcomes are generated based on cutoffs $\kappa$.
In this simulation, I consider three cases: $J \in \{3, 5, 7\}$.
For $J = 3$ and $J= 5$, the same cutoffs as in the previous simulation are used.
For $J = 7$, I use $\kappa = (-0.5, -0.2,  0.1,  0.4,  0.7,  1.0)^{\top}$.

\begin{figure}[htb]
  \centerline{\includegraphics[scale=0.9]{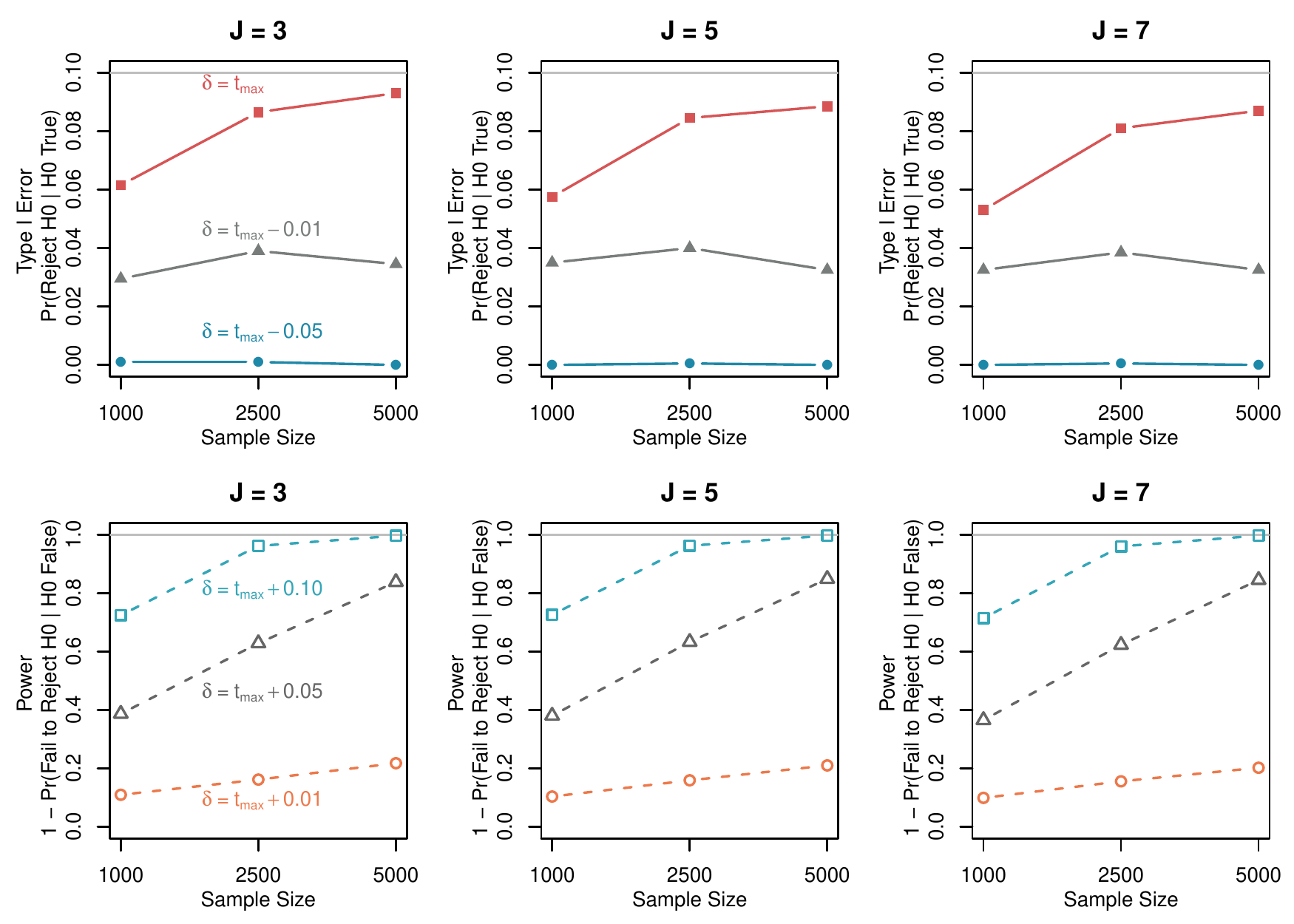}}
  \caption{Finite sample performance of the proposed testing procedure: Type I error (upper panel) and power (lower panel).
  The upper panel shows rejection probabilities of $H_{0}$ under thresholds that are compatible with $H_{0}$ (i.e., $H_{0}$ is true: $t_{\max} \geq \delta$ holds).
  The lower panel shows the power curve in a range of $\delta$ that is not compatible with $H_{0}$ (i.e., $H_{1}$ is true: $t_{\max} < \delta$).}
  \label{fig:simulation-equivalence}
\end{figure}

In order to assess how the test performs depending on a choice of equivalence thresholds,
I vary $\delta$.
The value of $\delta$ is chosen  such that  in some range of $\delta$,
the null of $t_{\max} \geq \delta$ is true
and in other range of $\delta$ the null is false (i.e., $t_{\max} < \delta$).
For the range of $\delta$ that satisfies $t_{\max} \geq \delta$,
I set $\delta \in \{t_{\max} - 0.05, t_{\max} - 0.01, t_{\max}\}$.
We would expect that the test is more likely to reject the alternative when $\delta = t_{\max} - 0.05$.
For the range of $\delta$ that does not satisfy $t_{\max} \geq \delta$,
I set $\delta \in \{t_{\max}+0.05, t_{\max}+0.01, t_{\max}+0.10\}$.
Among them, we expect that the test can reject the null most likely when $\delta = t_{\max}+0.10$.

Figure~\ref{fig:simulation-equivalence} shows the results for this simulation.
The upper panels show type I errors
when the choice of $\delta$ is consistent with the data (i.e., $t_{\max} \geq \delta$).
Recall that the data is simulated such that the equivalence does not hold.
Thus, we would expect that the null, $H_{0}\colon t_{\max} \geq \delta$, is not rejected
and the probability of  falsely rejecting the null (type I error) should be less than $\alpha$.
In fact, we can see that the proposed testing procedure controls the type I error.
In addition, the smaller value of $\delta$ (i.e., a smaller rejection region for $H_{0}$) leads to lower type I error.
The lower panels show results for type II errors
when the choice of $\delta$ is not consistent with the data (i.e., $H_{0}$ is false).
We can see that the test struggles to reject the null when $\delta$ is set to close to $t_{\max}$.
When $\delta$ is set to a value far away from $t_{\max}$, type II error converges to zero as sample size increases.

\end{document}